\documentclass[oneside,
bibliography=totoc,%
]{scrartcl}
\pdfoutput=1

\usepackage[a4paper]{geometry}
\usepackage[T1]{fontenc}
\usepackage[utf8]{inputenc}
\usepackage{lmodern}

\usepackage[dvipsnames]{xcolor}

\usepackage{amsmath,amssymb}
\usepackage{amsthm}
\usepackage{dsfont}
\usepackage{mathtools}
\usepackage{tikz,tikz-3dplot}
\usetikzlibrary{patterns,fadings,decorations.pathreplacing}

\usepackage{hyperref}

\usepackage{subcaption}
\usepackage{algorithm,algpseudocode}

\newcounter{parentnumber}
\makeatletter
\newenvironment{subtheorem}[1]{%
  \def\subtheoremcounter{#1}%
  \refstepcounter{#1}%
  \protected@edef\theparentnumber{\csname the#1\endcsname}%
  \protected@edef\theHparentnumber{\@ifundefined{theH#1}%
  \csname the#1\endcsname\csname theH#1\endcsname}%
  \setcounter{parentnumber}{\value{#1}}%
  \setcounter{#1}{0}%
  \expandafter\def\csname the#1\endcsname{\theparentnumber.\Alph{#1}}%
  \expandafter\def\csname theH#1\endcsname{\theparentnumber.\Alph{#1}}%
  \ignorespaces
}{%
  \setcounter{\subtheoremcounter}{\value{parentnumber}}%
  \ignorespacesafterend
}
\makeatother

\newtheorem{theorem}{Theorem}%
\newtheorem{lemma}[theorem]{Lemma}
\newtheorem{corollary}[theorem]{Corollary}
\newtheorem{proposition}[theorem]{Proposition}

\theoremstyle{definition}
\newtheorem{definition}[theorem]{Definition}

\theoremstyle{remark}
\newtheorem{remark}[theorem]{Remark}

\newcommand{\bb}[1]{{\boldsymbol{#1}}}

\newcommand{\tr}{\mathrm{tr}}
\newcommand{\inner}[2]{{{#1} \cdot {#2}}}

\renewcommand{\P}{\mathbb{P}}
\newcommand{\R}{\mathbb{R}}

\newcommand{\N}{\mathbb{N}}
\newcommand{\C}{\mathbb{C}}
\newcommand{\Z}{\mathbb{Z}}
\newcommand{\E}{\mathbb{E}}
\newcommand{\one}{\mathds{1}}
\DeclareMathOperator{\supp}{supp}

\DeclareMathOperator{\relint}{relint}
\DeclareMathOperator{\NP}{\mathcal{N}}
\DeclareMathOperator{\Tr}{Tr}
\newcommand{\dd}{\text{d}}
\DeclareMathOperator{\faces}{faces}
\DeclareMathOperator{\var}{Var}

\newcommand{\bigO}{\mathcal{O}}

\newcommand{\vrt}{\mathrm{vert}}

\DeclareMathOperator{\Aut}{Aut}

\DeclareMathOperator{\Exp}{Exp}

\renewcommand{\Re}{\operatorname{Re}}

\makeatletter
 \tikzoption{canvas is plane}[]{\@setOxy#1}
 \def\@setOxy O(#1,#2,#3)x(#4,#5,#6)y(#7,#8,#9)%
   {\def\tikz@plane@origin{\pgfpointxyz{#1}{#2}{#3}}%
    \def\tikz@plane@x{\pgfpointxyz{#4}{#5}{#6}}%
    \def\tikz@plane@y{\pgfpointxyz{#7}{#8}{#9}}%
    \tikz@canvas@is@plane
   }
 \makeatother

\hypersetup{
  pdfauthor={Michael Borinsky},
  pdftitle={Tropical Monte Carlo quadrature for Feynman integrals},
  pdfsubject={},
  pdfkeywords={},
  linkcolor  = RedViolet!85!black,
  citecolor  = BlueGreen!85!black,
  urlcolor   = YellowOrange!85!black,
  colorlinks = true,
}

\date{}

\begin{document}

\title{Tropical Monte Carlo quadrature for Feynman integrals}
\author{Michael Borinsky\thanks{\textsc{Nikhef, Science Park 105, Amsterdam 1098 XG, The Netherlands --- Preprint nr.~2020-027}}
}

\maketitle

\begin{abstract} 
We introduce a new method to evaluate algebraic integrals over the simplex numerically. This new approach employs techniques from tropical geometry and exceeds the capabilities of existing numerical methods by an order of magnitude.  The method can be improved further by exploiting the geometric structure of the underlying integrand. As an illustration of this, we give a specialized integration algorithm for a class of integrands that exhibit the form of a generalized permutahedron. This class includes integrands for scattering amplitudes and parametric Feynman integrals with tame kinematics. A proof-of-concept implementation is provided with which Feynman integrals up to loop order $17$ can be evaluated.
\end{abstract}

\setcounter{tocdepth}{1}
\tableofcontents
\clearpage

\section{Introduction}
\subsection{Motivation}

Feynman integrals are ubiquitous in various branches of theoretical physics. They are hard to evaluate and predictions for particle physics experiments rely heavily on them. Their evaluation even poses a bottleneck for the analysis of the data from some high accuracy experiments \cite{Heinrich:2020ybq}. This situation has fostered the development of extremely sophisticated and specialized technologies aimed to obtain a manageable analytic expression for a given Feynman integral. The state of the art technique is the \emph{differential equation method} \cite{Kotikov:1990kg,Remiddi:1997ny,Henn:2013pwa}. A slightly less powerful method, which is amendable to more general algebraic integrals, is systematic algebraic integration \cite{Brown:2008um,Panzer:2014caa}. See also \cite{Smirnov:2012gma} for an overview on other methods.

The rapid development of these technologies in the last decades has been driven largely by new deep insights into the underlying mathematical structures. 
Recent advances in the differential equation method were inspired by the simple analytic expressions which can be obtained in supersymmetric quantum field theories via generalized unitarity and recursion relations \cite{Bern:1994zx,ArkaniHamed:2010kv,Drummond:2010cz}.
A program to study the arithmetic properties of parametric Feynman integrals \cite{Bloch:2005bh,Brown:2009ta,Brown:2015fyf} led to the development of systematic algebraic integration algorithms. This arithmetic understanding of the relevant function classes was also one of the driving forces of the differential equation method \cite{henn2015lectures} and is still driving new developments in the especially challenging elliptic regime \cite{Broedel:2018iwv}.

All these technologies aim to obtain a closed form analytic expression for the Feynman integral and they all fail once the underlying graph and the associated physical parameters exceed a certain complexity. In these cases a numerical approach is the only way to proceed \cite{Borowka:2018dsa}.

The most established numerical approach to tackle such integrals is \emph{sector decomposition}. Sector decomposition as a tool for numerical evaluation of Feynman integrals has been developed by Binoth and Heinrich \cite{Binoth:2000ps}. It was subsequently improved by Bogner and Weinzierl \cite{Bogner:2007cr}. Another conceptual innovation of the overall method was achieved by Kaneko and Ueda \cite{Kaneko:2009qx} who brought sector decomposition on a geometric footing. Today, geometric sector decomposition is still the most powerful method for the numerical evaluation of Feynman integrals. It lies at the heart of two popular software tools \cite{Borowka:2015mxa,Smirnov:2013eza}. 
Another promising numerical technique for Feynman integration is \emph{loop-tree duality} \cite{Catani:2008xa}, which is in an active development phase (see for instance \cite{Runkel:2019yrs,Capatti:2019ypt} and the references therein).

In contrast to analytic evaluation methods the mathematical structures exhibited by Feynman integrals are an essentially untapped resource in the context of numerical evaluation. Most numerical techniques are completely oblivious to the rich specific structure of the integrals as they are designed to be applicable to arbitrary algebraic integrands.
For this reason, the major objective of this paper is to use some of these mathematical structures to improve the numerical evaluation of Feynman graphs and to show that these dormant resources can be harnessed.
The overall endeavour behind this objective consists of making progress towards the following two goals:

The first goal is to make numerical evaluation techniques more applicable to real world phenomenology. There are integrals which contribute to interesting measurable processes, but cannot be calculated analytically with available methods. For these integrals numerical evaluation is currently the only way to make predictions for experiments. 
Numerical methods naturally come with a caveat: they suffer from long evaluation times or they are limited in accuracy. 
Feynman integrals usually need to be evaluated a large number of times in a big parameter space. This is not difficult if an analytic expression for the integral is known, which can be evaluated sufficiently fast, but poses a tough challenge for numerical methods which sacrifice evaluation speed and accuracy for generality. The task for this goal is therefore to increase the performance of numerical methods. 

The second goal is to obtain reliable data in the \emph{large-order regime} where analytic methods hopelessly fail. There are many indications that the large-order behaviour of perturbation theory is deeply intertwined with non-perturbative phenomena \cite{le2012large}. The analysis of the large-order behaviour of perturbation theory in quantum mechanics by Bender and Wu \cite{bender1969anharmonic} has sparked an extremely fruitful branch of research in theoretical and mathematical physics \cite{AIHPA_1999__71_1_1_0,le2012large}. Non-perturbative analytic calculations in quantum field theory are plagued with various gaps in our understanding of the underlying mathematics \cite{McKane:2018ocs}. A repetition of an explicit Bender-Wu like numerical analysis for perturbative quantum field theory is very desirable as it would shed some light into a highly unexplored territory. Unfortunately, this is extremely challenging as the evaluation of large numbers of Feynman integrals of order $\sim 100$ would be necessary. 
It is hopeless to approach this task using the naive method of evaluation Feynman integrals one by one. New methods with which whole classes of diagrams can be evaluated at once need to be developed. The growing understanding of the geometry of amplitude integrals could lead the way to such methods. It is also necessary that these methods are \emph{computationally efficient}: The demand for computing resources shall at most depend polynomially on the size of the problem (i.e.~the respective order in the perturbative expansion).

This paper achieves some progress towards both these goals.  The strategy is to employ \emph{tropical geometry} \cite{maclagan2015introduction} for numerical quadrature. Panzer \cite{Panzer:2019yxl} recently showed that a tropical version of a Feynman graph's \emph{period} behaves similar to the period itself and anticipated that this tropical version may be used for explicit numerical evaluation. Tropical geometry has also recently been used in the context of string theory and scattering amplitudes \cite{Cachazo:2019ngv,Arkani-Hamed:2019mrd}. 

We will introduce a new Monte Carlo algorithm with which the numerical evaluation of Feynman integrals can be significantly accelerated. It is based on the established geometric sector decomposition principle and can be applied to general algebraic integrals such as the one below in eq.~\eqref{eq:integral}. The improvement comes from a stratified sampling approach to Monte Carlo quadrature, which is driven by the (tropical) geometric structure of the algebraic integrand. We will call this method \emph{tropical sampling}. Tropical sampling effectively decouples the complexity of the underlying integral from the achievable accuracy within the Monte Carlo approach. 

Even though this new tropical-geometric technique offers a significant improvement over the traditional procedure already for general algebraic integrals, a lot more can be achieved if further information on the geometric structures of the integral is used. 

As an example of this, we will give a specialized algorithm for cases where the integrand at hand exhibits the form of a \emph{generalized permutahedron} \cite{postnikov2009permutohedra} in a certain sense which will be defined later. Many integrals in quantum field theory and string theory fall under this category. For instance, integrands for complete amplitudes in various theories \cite{Arkani-Hamed:2017mur} and Feynman integrals with generic Euclidean kinematics are of this kind \cite{schultka2018toric}. 

A proof-of-concept implementation of the resulting algorithm is provided. With this implementation high dimensional Feynman integrals can be numerically integrated using widely available hardware. High dimensional explicitly means that integrals corresponding to Feynman graphs with around $20$ edges can be estimated up to $10^{-3}$ relative accuracy in a couple of CPU-seconds and integrals for graphs with up to $30$ edges in about half an hour. Ultimately, this approach is not CPU but memory constrained when the complexity increases. For instance, for a graph with $30$ edges already $16 \text{ GB}$ of computer memory are required to run the algorithm. To integrate an $18$ loop $\varphi^4$-theory four-point graph with the implementation $1 \text{ TB}$ of memory would be necessary. 

Although both algorithms are not \emph{efficient} in the strong sense, as exponential runtime and memory requirements start to dominate at some point, there is hope for the existence of an algorithm that evaluates a Feynman graph of loop order $n$ up to a given accuracy with runtime and memory demands bounded by a polynomial in $n$ (see Section~\ref{sec:openquestions} (1)).

\subsection{Algebraic integrals over the simplex}
The central object of study in this article is the integral
\begin{align} \label{eq:integral} I = \int_{\P_{>0}^{n-1}} \frac{\prod_{i} a_{i}(\bb{x})^{\nu_i}}{\prod_{j} b_{j}(\bb{x})^{\rho_j}} \Omega \quad \text{ with } \end{align}
\begin{itemize}
\item
the positive orthant of real projective space as integration domain,
\begin{align*} \P_{>0}^{n-1} = \{\bb x= [x_1:\ldots:x_n] \in \P^{n-1}(\R): x_k > 0 \text{ for all } k =1,\ldots,n\}, \end{align*}
\item the differential form 
\begin{align*} \Omega = \sum_{k=1}^n (-1)^{n-k} \frac{\dd x_1}{x_1} \wedge \ldots \wedge \widehat{\frac{\dd x_k}{x_k} } \wedge \ldots \wedge \frac{\dd x_n}{x_n}, \end{align*}
\item the sets of homogeneous polynomials $\{a_1,a_2,\ldots\},\{b_1,b_2,\ldots\} \subset \C[x_1, \ldots, x_n]$, s.t. 
\begin{align} \label{eq:homogeneous} \sum_i \nu_i \deg a_i = \sum_j \rho_j \deg b_j, \end{align}
\item where the coefficients $\nu_i,\rho_j \in \C$ have non-negative real part: $\Re \nu_i, \Re \rho_j \geq 0$ and
\item 
a fixed branch choice for the each of the non-integer powers, e.g.\ $a_i(\bb x)^{\nu_i} \mapsto e^{\nu_i \log a_i(\bb x)}$.
\end{itemize}
It follows that $I$ is a \emph{projective integral} over the \emph{projective simplex}. The differential form $\Omega$, which is homogeneous degree $0$, is also called the \emph{canonical form} on this simplex \cite{Arkani-Hamed:2017tmz}.  

Parametric Feynman integrals in quantum field theory are of the same type as the integral in eq.~\eqref{eq:integral} \cite{Nakanishi:110324}. In string theory this type of integrals plays an equally  important role \cite{Green:1987sp}. Integrals over \emph{positive geometries}, which appear in the theory of scattering amplitudes can also be brought into this form \cite{ArkaniHamed:2010kv,Arkani-Hamed:2017mur,Arkani-Hamed:2017tmz}.

The integral in eq.~\eqref{eq:integral} can be written as an integral over the positive orthant of $\R^n$: $\R^{n-1}_{>0} = \{(x_1,\ldots,x_{n-1})\in \R^{n-1}: x_k > 0\}$ by picking an \emph{affine chart} for projective space, for instance 
\begin{align} \label{eq:integral_euler_mellin} I = \int_{\R_{>0}^{n-1}} \left. \frac{\prod_{i} a_{i}(\bb{x})^{\nu_i}}{\prod_{j} b_{j}(\bb{x})^{\rho_j}} \right|_{x_n=1} \prod_{k=1}^{n-1} \frac{\dd x_k}{x_k}, \end{align}
by pulling back the diffeomorphism $\R^{n-1}_{>0} \rightarrow \P^{n-1}_{>0}, (x_1,\ldots,x_{n-1}) \mapsto [x_1, \ldots, x_{n-1}, 1]$.

Such an integral is called a \emph{generalized Euler-Mellin integral}. Continuing a program started by Gelfand, Kapranov and Zelevinsky (GKZ) \cite{gelfand1990generalized}, these integrals have been studied extensively by Nilsson and Passare and others \cite{nilsson2013mellin,berkesch2014euler}. 
This analysis is compatible with the geometric sector decomposition approach as was shown by Schultka \cite{schultka2018toric}, who studied these integrals using \emph{toric geometry}.  Recently, generalized Euler-Mellin integrals gained new attention in the context of positive geometries, scattering amplitudes and string theory \cite{Arkani-Hamed:2019mrd,He:2020onr}. Along these lines also the analysis of Feynman integrals as GKZ-type hypergeometric functions has recently gained a lot of attention \cite{delaCruz:2019skx,Klausen:2019hrg,Feng:2019bdx}. 

As every non-homogeneous polynomial in $n-1$ variables can be homogenized by introducing a new variable, every generalized Euler-Mellin integral such as the one in eq.~\eqref{eq:integral_euler_mellin} is equivalent to an integral of the projective form in eq.~\eqref{eq:integral}. 
For our considerations it will be more convenient to work with the projective form.

\subsection{Outline of the paper}
After introducing the necessary preliminaries from polyhedral geometry and numerical integration in Section~\ref{sec:preliminaries}, we will establish the most important tool in this paper in Section~\ref{sec:trop_approx}: 
an approximation of a multivariate polynomial, which is obtained by setting all its coefficients to $1$ and replacing $+$ by $\max$. Starting for instance with the polynomial $p(x_1, x_2, x_3) = a x_1^2 x_2 + b x_1 x_2 x_3 + c x_3^3$, we get the `approximation' $p^\tr(x_1,x_2,x_3) = \max(x_1^2 x_2, x_1 x_2 x_3, x_3^3)$. 
This procedure is inspired from and closely related to \emph{tropical geometry}. Hence, $p^\tr$ will be called the \emph{tropical approximation} of $p$. 
This tropical approximation will be the subject of the main Theorem~\ref{thm:approx} of this article, where it will be proven that $p^\tr$ can indeed by used to approximate the polynomial $p$ in a certain sense, as long as $p$ is \emph{completely non-vanishing}. In the rest of the article we will apply this property in various contexts.

In Section~\ref{sec:tropsecdec}, we will reformulate Kaneko-Ueda geometric sector decomposition in a tropical geometric framework. This reformulation will enable us to introduce the new tropical sampling algorithm in Section~\ref{sec:trop_sampling}. This algorithm is significantly more efficient than traditional Monte Carlo methods as the achievable accuracy is effectively decoupled from the complexity of the integral. Only the runtime of a preprocessing step still depends heavily on the complexity of the integral. 

This new algorithm can be improved further if more is known about the structure of the integrand polynomials $\{a_i\}$ and $\{b_j\}$. As an example for this, we will specialize to the case where the Newton polytopes of these polynomials are \emph{generalized permutahedra} in Section~\ref{sec:genperm}. Generalized permutahedra are a family of polytopes with a rich combinatorial structure. Many polytopes are from this family including \emph{associahedra}, the Newton polytopes of Symanzik polynomials and other polytopes at play in the theory of scattering amplitudes \cite{Arkani-Hamed:2017mur}. We will use results by Postnikov \cite{postnikov2009permutohedra}, Aguiar, Ardila \cite{aguiar2017hopf} and Fujishige, Tomizawa \cite{fujishige1983note} to formulate a specialized algorithm. This second new algorithm has more favorable runtime and memory requirements and is easier to implement. 

Even though the improved algorithm can be applied to all generalized permutahedra integrands, as for instance the ones for complete scattering amplitudes introduced by Arkani-Hamed, Bai, He and Yan \cite{Arkani-Hamed:2017mur}, we will specify to Feynman integrals in Section~\ref{sec:feynman} for illustrative purposes. 

The first, general tropical sampling algorithm can always be applied to Feynman integrals regardless of their explicit form. The second algorithm can only be applied if the Newton polytopes of the Symanzik polynomials are generalized permutahedra. We will use results by Brown \cite{Brown:2009ta} and Schultka \cite{schultka2018toric} which ensure that this is the case as long as we are in a non-exceptional Euclidean kinematic region. Subsequently, we will discuss some experimental results which have been obtained using a proof-of-concept implementation of the second algorithm. 

We conclude with a selection of future research directions resulting from this project in the last Section~\ref{sec:openquestions}.

\subsection*{Acknowledgements}
I am indebted to David Broadhurst, Iain Crump, Alejandro Morales and Erik Panzer for helpful discussions during the workshops at SFU in 2016 and at the University of Waterloo in 2018 both organized by Karen Yeats, the Programme on `Algorithmic and Enumerative Combinatorics' at the Erwin-Schrödinger Institute in 2017 and the `Summer school on structures in local quantum field theory' at the École de physique des Houches in 2018 organized by Dirk Kreimer.
This work also greatly benefited from discussions with Francis Brown, Gudrun Heinrich, Franz Herzog, Dirk Kreimer, Oliver Schnetz, Konrad Schultka, Jos Vermaseren and Karen Yeats.
I also wish to thank Erik Panzer for comments and suggestions on an early version of this manuscript and Jos Vermaseren for granting me generous access to his computing resources.

This work has been supported by the NWO Vidi grant 680-47-551 `Decoding Singularities of Feynman graphs'.
\section{Preliminaries}
\label{sec:preliminaries}
\subsection{Notation for polytopes and multivariate polynomials}
The integral in eq.~\eqref{eq:integral} is convergent if the sets of polynomials $\{a_i\}$ and $\{b_j\}$ fulfill certain properties, which essentially have been determined by Nilsson and Passare \cite{nilsson2013mellin}. In this section, we will briefly review these properties and introduce the necessary vocabulary from polyhedral and tropical geometry.

To keep the notation simple, we will identify the space of linear forms on $\R^n$ with $\R^n$ via the usual scalar product $\inner{\bb v}{\bb w} = \sum_{k=1}^n v_k w_k$. A polytope is the intersection of a finite number of \emph{half-spaces} in $\R^n$. A subset $F \subset \mathcal P$ of a polytope $\mathcal P \subset \R^n$ is a face of $\mathcal P$ if there is a vector $\bb y \in \R^n$ and a scalar $\xi \in \R$ such that $\mathcal P$ is contained in the half-space $\{\bb v\in \R^n: \inner{\bb y}{\bb v} \leq \xi \}$ and $F$ is the intersection of $\mathcal P$ with the \emph{hyperplane} $\{\bb v\in \R^n: \inner{\bb y}{\bb v} = \xi \}$. Equivalently, a face of a polytope is a subset of $\mathcal P$ which maximizes a given linear functional $\bb y \in \R^n$, $F = \{ \bb v \in \mathcal P: \inner{\bb y}{\bb v} = \max_{\bb w \in \mathcal P} \inner{\bb y}{\bb w} \}$. We will assume polytopes to be bounded, i.e.~$\max_{\bb w \in \mathcal P} \inner{\bb y}{\bb w} < \infty$ for all $\bb y\in \R^n$.

For a pair of non-negative real numbers $\lambda, \mu \geq 0$ the weighted Minkowski sum of two polytopes $\mathcal{P}, \mathcal{Q} \subset \R^n$ is $\lambda \mathcal{P} + \mu \mathcal{Q} = \left\{ \lambda \bb{v} + \mu \bb{w} : \bb{v} \in \mathcal{P}, \bb{w} \in \mathcal{Q} \right\} \subset \R^n$.  The relative interior, $\relint(\mathcal P)$, of a polytope $\mathcal P$ is the interior of $\mathcal P$ determined in the subtopology of the affine hull of $\mathcal P$, which is the affine subspace of minimal dimension that contains $\mathcal P$.

We can write a generic multivariate polynomial $p$ in the variables $x_1, \ldots, x_n$ as
\begin{align*} p(x_1,\ldots, x_n) = p(\bb x) = \sum_{\bb{\ell} \in \supp(p)} c_{\bb{\ell}} \prod_{k=1}^n x_k^{\ell_k} = \sum_{\bb{\ell} \in \supp(p)} c_{\bb{\ell}} \bb x^{\bb \ell}, \end{align*}
where $\supp(p)$, the \emph{support} of $p$, is the set of all multi-indices $(\ell_1, \ldots, \ell_n) = \bb{\ell} \in \Z^{n}$ such that $c_{\bb{\ell}} \neq 0$. We will make regular use of the multiplicative multi-index notation $\bb x^{\bb \ell} = \prod_{k=1}^n x_k^{\ell_k}$ as above.
The Newton polytope of $p$ is the convex hull of the elements in $\supp(p)$ interpreted as vectors in $\R^n$:
\begin{align*} \NP_{p}= \left\{ \sum \limits_{\bb{\ell} \in \supp(p)} \lambda_{\bb{\ell}} \bb{\ell} : \sum \limits_{\bb{\ell} \in \supp(p)} \lambda_{\bb{\ell}} = 1 \text{ and } \lambda_{\bb{\ell}} \geq 0 \right\} \subset \R^n. \end{align*}
The Newton polytope $\NP_p$ of a homogeneous polynomial $p$ in $n$ variables is at most ($n-1$)-dimensional as it is contained in the hyperplane $\{\bb{v} \in \R^n : \inner{\one}{\bb{v}} = \deg{p}\} \supset \NP_p$, where $\one$ is the only-ones-vector $\one=(1,\ldots, 1)\in \R^n$.

\begin{definition}
\label{def:truncated_polynomial}
For each face $F$ of $\NP_p$ associated to a polynomial $p$, we define the truncated polynomial $p_F$ by
\begin{align*} p_F(\bb x) = \sum_{\bb{\ell} \in F \cap \supp(p)} c_{\bb{\ell}} \bb x^{\bb \ell}. \end{align*}
\end{definition}

To ensure convergence of integrals such as the one in eq.~\eqref{eq:integral}, the following property of polynomials is useful:
\begin{definition}
A polynomial $p \in \C[x_1, \ldots, x_n]$ is \emph{completely non-vanishing} on a domain $X$ if for each face $F \subset \NP_p$, the truncated polynomial $p_F$ does not vanish on $X$.
\end{definition}

With this terminology at hand we can give a convergence criterion for the integral in eq.~\eqref{eq:integral}.
\begin{theorem}
\label{thm:convergence}
We define the polytopes $\mathcal A,\mathcal B \subset \R^n$ as  the weighted Minkowski sums
\begin{align*} \mathcal{A} &= \sum_{i} (\Re \nu_i) \NP_{a_i} & \mathcal{B} &= \sum_{j} (\Re \rho_j) \NP_{b_j} \end{align*}
of the Newton polytopes of the numerator and denominator polynomials $\{a_i\}$ and $\{b_j\}$.

The integral in eq.~\eqref{eq:integral} is convergent if
\begin{enumerate}
\item[\normalfont \textbf{R1}]
the denominator polytope $\mathcal B$ is $(n-1)$-dimensional,
\item[\normalfont \textbf{R2}]
the numerator polytope $\mathcal A$ is contained in the relative interior of $\mathcal B$: $\mathcal A \subset \relint \mathcal B$,
\item[\normalfont \textbf{R3}]
all the denominator polynomials $\{b_j\}$ are completely non-vanishing on $\P^{n-1}_{>0}$.
\end{enumerate}
\end{theorem}

\begin{remark}
\label{rmk:homogeneous}
Because each of the $\{a_i\}$ and $\{b_j\}$ polynomials is homogeneous, neither $\mathcal A$ nor $\mathcal B$ is full dimensional in $\R^n$. 
The condition in eq.~\eqref{eq:homogeneous}, which implies that $\sum_i \Re \nu_i \deg a_i = \sum_j \Re \rho_j \deg b_j$, guarantees that $\mathcal A$ and $\mathcal B$ both lie in the same hyperplane $\mathcal A, \mathcal B \subset \{\bb{v} \in \R^n : \inner{\one}{\bb{v}} = \xi\}$, where $\xi = \sum_i \Re \nu_i \deg a_i = \sum_j \Re \rho_j \deg b_j$ in which $\mathcal B$ is required to be full-dimensional by requirement \normalfont{R1}.
\end{remark}

A similar theorem in the equivalent context of Euler-Mellin integrals was proven in \cite{nilsson2013mellin} (see also \cite{berkesch2014euler}). The \emph{tropical approximation} that we will introduce later will lead to an alternative proof of Theorem~\ref{thm:convergence} which we postpone to Section~\ref{sec:tropsecdec}. %
\subsection{Monte Carlo quadrature}
We will be interested in situations where the dimension $n$ of the integral in eq.~\eqref{eq:integral} is `not small'. The dimension of an integral is small from the perspective of numerical quadrature if fast-converging deterministic quadrature methods are feasible. The computational demands of deterministic quadrature methods such as Gauss-quadrature grow exponentially with the dimension. For this reason, it is necessary to use non-deterministic methods which do not suffer from an exponential slow-down if $n$ is not small. Monte Carlo quadrature is the most elementary of these. 
The working principle behind it is the following fact:
\begin{theorem}[Monte Carlo quadrature (see for instance \cite{hammersley1964monte})]
\label{thm:montecarlo}
If $\bb x^{(1)},\ldots,\bb x^{(N)}$ are independent random variables with probability density measure $\mu$, i.e.\ $1 = \int_\Gamma \mu$ and $\mu>0$ on the domain $\Gamma$ and
\begin{gather*} G^{(N)} = \frac{1}{N} \sum_{\ell=1}^N f(\bb x^{(\ell)}), \intertext{then} \E[ G^{(N)} ] = \E [ f(\bb x) ] = \int_\Gamma f(\bb x) \mu \\ \text{and } \var[ G^{(N)} ] = \frac{1}{N} \var[ f(\bb x) ] \text{ where } \var[f(\bb x)] = \int_\Gamma | f(\bb x) - \E[ f(\bb x) ] |^2 \mu, \end{gather*}
provided that the integrals in the last two lines exist.
\end{theorem}
This theorem may be applied to approximate the integral $\int_\Gamma f(\bb x) \mu$ as long as we have a way of generating samples from the distribution $\mu$.  The condition that the integral for the variance shall exist effectively restricts the set of functions $f$, which can be integrated numerically, to the set of \emph{square integrable functions} under the measure $\mu$ over the domain $\Gamma$. Note that in our convention of the statement of Theorem~\ref{thm:montecarlo} the expectation value $\E[ \cdot ]$ may be complex, but the variance $\var[ \cdot]$ is always real and non-negative.

The integral in eq.~\eqref{eq:integral} is not directly amendable to Monte Carlo quadrature as the differential form $\Omega$ over the domain $\P_{>0}^{n-1}$ as defined for eq.~\eqref{eq:integral} is not a probability distribution: it is not normalizable. Even if we use the affine representation in eq.~\eqref{eq:integral_euler_mellin} and map $\R_{>0}^n$ onto a bounded domain via a variable transformation (for instance by $x \mapsto x/(1+x)$ which maps $\R_{>0}\rightarrow (0,1)$ smoothly and injectively), the resulting integral will, in the general case, not be square integrable. A pragmatic solution to this problem is \emph{sector decomposition} \cite{Binoth:2000ps}, where the integral $I$ is expressed as a sum of integrals, which are each individually directly amendable to Monte Carlo integration. 

\subsection{Sector decomposition}
\label{sec:secdec}

In the context of quantum field theory, sector decomposition goes back to Hepp and Speer who used the technique to prove the finiteness of renormalized Feynman integrals \cite{Hepp:1966eg,Speer:1975dc}. Even though Hepp/Speer sector decomposition can be employed to deal with the singularities of Euclidean Feynman integrals, it turned out to be insufficient to handle more general singularities which appear in Minkowski space Feynman integrals (see \cite{Heinrich:2008si} and \cite[Chapter~4]{Smirnov:2012gma} for reviews on the topic). 
A more general approach was pioneered by Binoth and Heinrich \cite{Binoth:2000ps}, who introduced a recursive algorithm that decomposes general integrals of the form in eq.~\eqref{eq:integral} (or equivalently as in eq.~\eqref{eq:integral_euler_mellin}) into a set of \emph{sector integrals}: 
\begin{align} \label{eq:secdec_integral} I &= \sum_{s \in S} I_s & I_s &= C_s \int_{[0,1]^{n-1}} \bb x^{\bb m^{(s)}} \frac{\prod_i {\widetilde p}^{\nu_i}_{s,i}(\bb x)}{\prod_j {\widetilde q}^{\rho_j}_{s,j}(\bb x)} \prod_{k=1}^{n-1} \frac{\dd x_k}{x_k}, \end{align}
such that the auxiliary polynomials ${\widetilde p}_{s,i}(\bb x)$ and ${\widetilde q}_{s,j}(\bb x)$ do not vanish inside the integration domain $[0,1]^{n-1}$ (at least as long as all the coefficients of the initial denominator polynomials $b_j$ are positive, which implies that these polynomials are completely non-vanishing) and $C_s$ is a prefactor for each sector $s\in S$. If all components of the vector $\bb m^{(s)}$ are positive, i.e.\ $m^{(s)}_k > 0$ for all $k\in\{1,\ldots, n-1\}$, a simple reparametrization $\bb \xi = \bb x^{\bb m^{(s)}}$ produces an integral with a bounded integrand 
\begin{align} \label{eq:secdec_integral_trafo} I_s = C_s \int_{[0,1]^{n-1}} \frac{\prod_i {\widetilde p}^{\nu_i}_{s,i}(\bb x(\bb \xi))}{\prod_j {\widetilde q}^{\rho_j}_{s,j}(\bb x(\bb \xi))} \prod_{k=1}^{n-1} \dd \xi_k, \end{align}
to which basic Monte Carlo as described in Theorem~\ref{thm:montecarlo} can immediately be applied using the uniform measure $\mu = \prod_{k=1}^{n-1} \dd \xi_k$ on the unit $(n-1)$-cube $\Gamma = [0,1]^{n-1}$. 

Although a problem of the method which impeded the recursion from terminating was solved by Bogner and Weinzierl \cite{Bogner:2007cr}, this class of algorithms suffers from a proliferation in the numbers of sectors with rising complexity of the underlying polynomials. Moreover, a new set of polynomials ${\widetilde p}_{s,i}, {\widetilde q}_{s,i}$ is associated to each sector. This means that we are not necessarily dealing with a partition of the integration domain, but a non-trivial distribution of the volume of the integral into each of the sectors $s\in S$. 

A both conceptual and practical innovation was achieved by Kaneko and Ueda, who reinterpreted this decomposition as a geometric problem \cite{Kaneko:2009qx}. This geometric viewpoint results in more economical decompositions, in terms of the total number of sectors (see \cite{Smirnov:2008aw} for a comparison of different methods), while also arguably being conceptually more elegant. %
\subsection{Analytic continuation}
By Theorem~\ref{thm:convergence}, the convergence of the integral in eq.~\eqref{eq:integral} depends on the values of the parameters $\{\nu_i\}$ and $\{\rho_j\}$.  Provided that there is an extended domain of such parameters where the integral is convergent, we can interpret it as a function of these parameters and perform an analytic continuation. It turns out that this analytic continuation is a \emph{meromorphic function} in these parameters \cite{nilsson2013mellin,berkesch2014euler}. 

A sector decomposition as in eq.~\eqref{eq:secdec_integral} provides a pragmatic way to perform this analytic continuation. 
A violation of the condition $\mathcal A \subset \relint \mathcal B$ in Theorem~\ref{thm:convergence} corresponds to a component of $\bb m^{(s)}$ in the integral in eq.~\eqref{eq:secdec_integral} being non-positive, i.e.\ $m^{(s)}_k \leq 0$ for some sector $s\in S$. Hence, if we assume that the polynomials ${\widetilde p}_{s,i}, {\widetilde q}_{s,i}$ are non-vanishing on the integration domain, then the associated integral $I_s$ is divergent. Performing an analytic continuation of this integral, interpreted as a function of the coefficients of $\bb m^{(s)}$ is a simple task. A standard approach is to integrate over a Pochhammer contour instead of the unit interval in eq.~\eqref{eq:secdec_integral_trafo}. This avoids the singularity at the integration boundary and agrees with the integral over the unit interval if convergence is ensured (see for instance \cite[Section~12-43]{whittaker1963course}).

For explicit computations it is sufficient to compute a Taylor expansion of the rational function in eq.~\eqref{eq:secdec_integral_trafo} up to an appropriate order, integrate the analytically continued expansion terms analytically and the remainder term numerically. See for instance \cite[Part~III]{Binoth:2000ps}, where this process is described in detail.

A more sophisticated strategy to perform this analytic continuation is based on iteratively performing `directed integration by parts' on the integral in eq.~\eqref{eq:integral} and thereby extending the domain of $\{\nu_i\}$, $\{\rho_j\}$ parameters in which the integral converges. The inner workings of this procedure are of geometric nature and make use of the structure of the polytopes $\mathcal A, \mathcal B$. See \cite[Theorem~2]{nilsson2013mellin} and thereafter and also \cite[Theorem~2.4]{berkesch2014euler} for a description of this method. This approach gives, after being applied to a given integral as the one in eq.~\eqref{eq:integral}, a sum of integrals of the same type where each integral has a larger region of convergence than the original one. A similar procedure has been developed independently in \cite{von2015quasi} for the special case of parametric Feynman integrals. 

In this work, we will therefore assume that the integral has been subjected to such a procedure and we can assume that we are within the region of convergence in terms of the $\{\nu_i\}, \{\rho_j\}$ parameters.
\section{The tropical approximation}
\label{sec:trop_approx}

For the considerations in this article the following \emph{tropical approximation} of a polynomial will be central:
\begin{definition} 
\label{def:trop}
For a polynomial $p \in \C[x_1, \ldots, x_n]$ define 
$p^\tr(\bb x) = \max \limits_{\bb{\ell} \in \supp(p)} \bb x^{\bb \ell}$.
\end{definition}
Such an object has been defined by Panzer~\cite{Panzer:2019yxl} for the Kirchhoff polynomial to study the \emph{Hepp-bound}, a graph invariant relevant for Feynman period integral calculations. We adopt Panzer's notation and denote tropically approximated polynomials with a superscript $^\tr$.

To give some additional motivation to consider this `tropical approximation' suppose that a polynomial $p$ has only real and positive coefficients and interpret it as a function $p:\R_{> 0}^n \rightarrow \R_{> 0}$. The \emph{tropical limit} is
\begin{align*} \lim_{\xi \rightarrow \infty} p(x_1^\xi,\ldots, x_n^\xi)^{\frac{1}{\xi}} = \lim_{\xi \rightarrow \infty} \left( \sum_{\bb{\ell} \in \supp(p)} c_{\bb{\ell}} \bb x^{\xi \bb \ell} \right)^{\frac{1}{\xi}} = \max_{\bb{\ell} \in \supp(p)} \bb x^{\bb \ell} = p^\tr(\bb x). \end{align*}
This way, $p^\tr$ can be seen as a \emph{deformed} version of $p$: the function $p(\bb x^\xi)^{\frac{1}{\xi}}$ interpolates between $p$ and $p^\tr$ with $\xi$ between $1$ and $\infty$. A limit as $\xi \rightarrow \infty$ with the associated phenomenon of transforming a very smooth object---in this case a polynomial---into a function with non-differentiable singularities, is something commonly encountered in physics. For instance, the \emph{thermodynamical limit} is of similar nature. These kind of limits give rise to numerous \emph{critical phenomena}. Also the weak string coupling limit $\alpha' \rightarrow 0$ shows this behavior \cite{Arkani-Hamed:2019mrd}.

In our case $p^\tr$ is of `simpler' nature as the original polynomial $p$. Information is lost while going from a polynomial to its tropical approximation, as $p^\tr$ only depends on the support of $p$. In fact, $p^\tr$ is nothing but a realization of a geometric object: the Newton polytope of the polynomial $p$.

To make this explicit, change to logarithmic coordinates $y_k = \log x_k$ in Definition~\ref{def:trop} and use the fact that the Newton polytope is the convex hull of the support of the underlying polynomial. We find that
\begin{align} \label{eq:logPtr} \log p^\tr(\bb{x}) = \max \limits_{\bb{\ell} \in \supp(p)} \log \bb x^{\bb \ell} = \max \limits_{\bb{\ell} \in \supp(p)} \sum_{k=1}^n y_k \ell_k = \max \limits_{\bb{\ell} \in \supp(p)} \inner{\bb{y}}{\bb{\ell}} = \max_{\bb{v} \in \NP_p} \inner{\bb{y}}{\bb{v}}, \end{align}
which is a piece-wise linear function $\R^n \rightarrow \R, \bb y \mapsto \max_{\bb{v} \in \NP_p} \inner{\bb{y}}{\bb{v}}$. 
This function is the \emph{support function} of the Newton polytope $\NP_p$ \cite{henk200416}.
Often it is useful to write $p^\tr$ in exponential form using the support function:
\begin{proposition}
\label{prop:trsupport}
$p^\tr(e^\bb{y}) = e^{\max_{\bb v \in \NP_p}\inner{\bb y}{\bb v}} $,
\end{proposition}%
\noindent
where we used the notation $e^\bb{y}=(e^{y_1}, \ldots, e^{y_n})$  to denote the component-wise exponential.

This support function is also a \emph{tropicalization} of the polynomial $p$, which uses the trivial valuation on $\C$ to tropicalize. Because much of the algebro-geometrical information of the polynomial $p$ carries over to its tropicalization, tropical geometry developed into a fruitful branch of algebraic geometry in the recent years \cite{maclagan2015introduction}.  
It is tempting to call $p^\tr$ the tropicalization of $p$. Unfortunately, this name is reserved for $\bb y \mapsto \log p^\tr(e^{\bb y})$ and we will use the name \emph{tropical approximation} instead. The motivation for this is that besides the fact that $p^\tr$ is a simplification of $p$, it can also be used to \emph{approximate} $p$.

\subsection{The approximation property}
The main theorem of this article is the following \emph{approximation property} of $p^\tr$ with respect to the polynomial $p$: 
\begin{subtheorem}{theorem}
\label{thm:approx}
\begin{theorem}
\label{thm:approx_upper}
For every polynomial $p \in \C[x_1, \ldots, x_n]$ there is a constant $C>0$ such that
\begin{align*} | p(\bb{x}) | &\leq C p^\tr (\bb{x}) \text{ for all } \bb{x} \in \R^{n}_{> 0}. \end{align*}
\end{theorem}
\begin{proof}
$C = \sum_{\bb \ell \in \supp(p)} |c_{\bb \ell}|$.
\end{proof}

\begin{theorem}
\label{thm:approx_lower}
If $p \in \C[x_1, \ldots, x_n]$ is completely non-vanishing on $\R^{n}_{>0}$, then there is a constant $C>0$ such that
\begin{align*} C p^\tr(\bb{x}) \leq | p(\bb{x}) | \text{ for all } \bb{x} \in \R^{n}_{> 0}. \end{align*}
\end{theorem}
\end{subtheorem}

This property trivially extends to homogeneous polynomials which are naturally considered to be functions on projective space:
\begin{corollary}
\label{crll:homo_approx}
For every homogeneous polynomial $p \in \C[x_1, \ldots, x_n]$ there is a constant $C>0$ such that
\begin{align*} | p(\bb{x}) | &\leq C p^\tr (\bb{x}) \text{ for all } \bb{x} \in \P^{n-1}_{> 0}. \end{align*}
and if $p$ is additionally completely non-vanishing on $\P^{n-1}_{>0}$, then there is a constant $C>0$ such that
\begin{align*} C p^\tr(\bb{x}) \leq | p(\bb{x}) | \text{ for all } \bb{x} \in \P^{n-1}_{> 0}. \end{align*}
\end{corollary}
Strictly speaking $p(\bb{x})$ and $p^\tr(\bb{x})$ are not well-defined objects for $\bb{x} \in \P^{n-1}_{> 0}$. The quotient $p(\bb{x})/p^\tr(\bb x)$, on the other hand, makes sense for all projective $\bb{x} \in \P^{n-1}_{> 0}$. Inequalities as the one above, which can be written as quotients of homogeneous objects, shall be interpreted accordingly, in this obvious sense as statements on these quotients. 
\begin{proof}
If $p$ is homogeneous and completely non-vanishing on $\P^{n-1}_{>0}$, it is also completely non-vanishing on $\R_{>0}^{n}$. The inequality in Theorems~\ref{thm:approx_upper} and \ref{thm:approx_lower} is homogeneous, therefore it trivially extends to $\P^{n-1}_{>0}$.
\end{proof}

By Theorem~\ref{thm:approx}, $p^\tr$ can indeed be used to `approximate' $p$, i.e.\ it provides a \emph{lower and an upper bound} of $p$ with appropriate prefactors, as long as $p$ is completely non-vanishing. 

Theorem~\ref{thm:approx_lower} is substantially harder to prove than Theorem~\ref{thm:approx_upper}. This proof of Theorem~\ref{thm:approx_lower} will be given in the next Section~\ref{sec:cones_and_fans}. Only a special case is also trivial: If the polynomial $p$ has only positive coefficients (which implies that $p$ is completely non-vanishing on $\R_{>0}^n$), there is a simple lower bound for $|p(\bb x)|$: take for instance $C = \min_{\bb \ell \in \supp(p)} c_{\bb \ell}$. 
For such a lower bound to exist it is not necessary for the polynomial to have only positive coefficients; it is sufficient for the polynomial to be completely non-vanishing. In fact, the existence of such a lower bound is also necessary for a polynomial to be completely non-vanishing, which can be proven using a similar argument as in the proof of Theorem~\ref{thm:approx_lower} below. 
\subsection{Cones and normal fans}
\label{sec:cones_and_fans}

A \emph{(polyhedral) cone} is a subset of $\R^n$ that is closed under linear combinations with only non-negative scalars, e.g.\ $\mathcal C = \{ \sum_{k} \lambda_k \bb u^{(k)} : \lambda_k \geq 0 \}$ for some set of given vectors $\bb u^{(1)}, \bb u^{(2)}, \ldots \in \R^n$.
A \emph{fan} in $\R^n$ is a family $\mathcal{F} = \{ \mathcal C_1, \mathcal C_2, \ldots \}$ of cones with the property that every face of a cone in $\mathcal F$ is also in $\mathcal F$ and that the intersection of two cones $\mathcal C_1, \mathcal C_2 \in \mathcal F$ is a face of both $\mathcal C_1$ and $\mathcal C_2$. The \emph{normal cone} associated to a face $F$ of the polytope $\mathcal P$ is the set of all linear functionals that are maximal on the respective face, 
\begin{gather} \label{eq:normal_cone} \mathcal C_F = \left \{ \bb y \in \R^n : \inner{\bb y}{\bb v} = \max_{\bb w \in \mathcal{P}} \inner{\bb y}{\bb w} \text{ for all } \bb v \in F \right\}         . \end{gather}
The set of normal cones of a polytope is its \emph{normal fan}: $\mathcal{F}_N = \{ \mathcal C_F : F \in \faces(\mathcal P), F \neq \emptyset \}$. The normal fan is always complete, that means $\R^n = \biguplus_{\mathcal C \in \mathcal F_N} \relint \mathcal C$, where $\uplus$ denotes the disjoint union. If a face $F$ has dimension $d$, then the associated normal cone $\mathcal C_F$ has dimension $n-d$, where $n$ is the dimension of the ambient space. The \emph{maximal cones} in the fan are the cones of maximal dimension. Figure~\ref{fig:polytope} and \ref{fig:normal_fan} depict a polytope and its normal fan. Note that the maximal cones can be associated to the vertices of the polytope.

For the proof of Theorem~\ref{thm:approx_lower}, it is convenient to have three further properties of the tropical approximation $p^\tr$ and the associated polytopes at hand:

\begin{lemma}
\label{lmm:truncated_homogenous}
For each face $F$ of the Newton polytope $\NP_p$ of a polynomial $p$, the truncated polynomial $p_F$ fulfills
\begin{align*} p_F(e^{\bb s+\bb t}) = p^\tr(e^{\bb s}) p_F(e^{\bb t}) \text{ for all } \bb s \in \mathcal C_F \text{ and } \bb t \in \R^n. \end{align*}
\end{lemma}
\begin{proof}
Use Definition~\ref{def:truncated_polynomial}, eq.~\eqref{eq:normal_cone} and Proposition~\ref{prop:trsupport}.
\end{proof}

\begin{lemma}
\label{lmm:submult}
$p^\tr(e^{\bb s+\bb t}) \leq p^\tr(e^{\bb s}) p^\tr(e^{\bb t}) \text{ for all } \bb s, \bb t \in \R^n$.
\end{lemma}
\begin{proof}
This follows from Proposition~\ref{prop:trsupport} and 
\begin{gather*} \max_{\bb v \in \mathcal P} \inner{(\bb s + \bb t)}{\bb v} \leq \max_{\bb v \in \mathcal P} \inner{\bb s}{\bb v} +\max_{\bb v \in \mathcal P} \inner{\bb t}{\bb v}. \qedhere \end{gather*}
\end{proof}

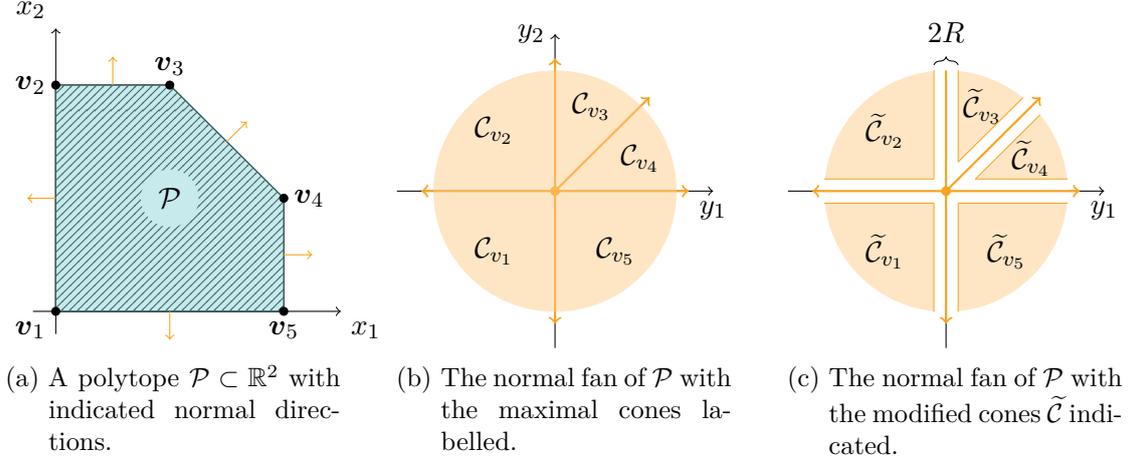
\begin{figure}
    \begin{subfigure}[t]{0.3\textwidth}
            \centering
        \begin{tikzpicture}[scale=1.5] \draw[thin,->] (-.2,0) -- (2.5,0) node[anchor=north west]{$x_1$}; \draw[thin,->] (0,-.2) -- (0,2.5) node[anchor=south east]{$x_2$}; \draw (0,0) -- (0,2) -- (1,2) -- (2,1) -- (2,0) -- (0,0); \draw[pattern=north east lines] (0,0) -- (0,2) -- (1,2) -- (2,1) -- (2,0) -- (0,0); \node(P) at (1,1) [circle,fill=white] {$\phantom{\mathcal P}$}; \filldraw[Aquamarine!50, opacity = 0.5] (0,0) -- (0,2) -- (1,2) -- (2,1) -- (2,0) -- (0,0); \node(P) at (1,1) [circle] {$\mathcal P$}; \draw[YellowOrange,->] (0,1) -- (-.25,1); \draw[YellowOrange,->] (.5,2) -- (.5,2.25); \draw[YellowOrange,->] (1.5,1.5) -- (1.6767766953,1.6767766953); \draw[YellowOrange,->] (2,.5) -- (2.25,.5); \draw[YellowOrange,->] (1,0) -- (1,-.25); \filldraw (0,0) circle(1pt) node[below left] {$\bb v_1$}; \filldraw (0,2) circle(1pt) node[left] {$\bb v_2$}; \filldraw (1,2) circle(1pt) node[above] {$\bb v_3$}; \filldraw (2,1) circle(1pt) node[right] {$\bb v_4$}; \filldraw (2,0) circle(1pt) node[below] {$\bb v_5$}; \end{tikzpicture}
    \caption{A polytope $\mathcal P \subset \R^2$ with indicated normal directions.}
    \label{fig:polytope}
    \end{subfigure}%
    \hfill
    \begin{subfigure}[t]{0.3\textwidth}
            \centering
        \begin{tikzpicture}[scale=1.6] \draw[thin,->] (-1.3,0) -- (1.3,0) node[anchor=north]{$y_1$}; \draw[thin,->] (0,-1.3) -- (0,1.3) node[anchor=east]{$y_2$}; \draw[thick,YellowOrange,->] (0,0) -- (-1.1,0); \draw[thick,YellowOrange,->] (0,0) -- (0,1.1); \draw[thick,YellowOrange,->] (0,0) -- (0.70710678118+0.070710678118,0.70710678118+0.070710678118); \draw[thick,YellowOrange,->] (0,0) -- (1.1,0); \draw[thick,YellowOrange,->] (0,0) -- (0,-1.1); \filldraw[YellowOrange] (0,0) circle(1pt); \filldraw[draw=none,fill = YellowOrange!50, opacity = .5] (0,0) -- ([shift={(180:1)}]0,0) arc (180:270:1) -- (0,0); \node(C1) at (-.5,-.5) [circle] {$\mathcal C_{v_1}$}; \filldraw[draw=none,fill = YellowOrange!50, opacity = .5] (0,0) -- ([shift={(90:1)}]0,0) arc (90:180:1) -- (0,0); \node(C2) at (-.5,.5) [circle] {$\mathcal C_{v_2}$}; \filldraw[draw=none,fill = YellowOrange!50, opacity = .5] (0,0) -- ([shift={(45:1)}]0,0) arc (45:90:1) -- (0,0); \node(C3) at (.3,.7) [circle] {$\mathcal C_{v_3}$}; \filldraw[draw=none,fill = YellowOrange!50, opacity = .5] (0,0) -- ([shift={(0:1)}]0,0) arc (0:45:1) -- (0,0); \node(C4) at (.7,.27) [circle] {$\mathcal C_{v_4}$}; \filldraw[draw=none,fill = YellowOrange!50, opacity = .5] (0,0) -- ([shift={(-90:1)}]0,0) arc (-90:0:1) -- (0,0); \node(C5) at (.5,-.5) [circle] {$\mathcal C_{v_5}$}; \end{tikzpicture}
    \caption{The normal fan of $\mathcal P$ with the maximal cones labelled.}
    \label{fig:normal_fan}
    \end{subfigure}%
    \hfill
    \begin{subfigure}[t]{0.3\textwidth}
            \centering
        \begin{tikzpicture}[scale=1.6] \draw[thin,->] (-1.3,0) -- (1.3,0) node[anchor=north]{$y_1$}; \draw[thin,->] (0,-1.3) -- (0,1); \filldraw[draw=none,fill = YellowOrange!50, opacity = .5] (0,0) -- ([shift={(180:1)}]0,0) arc (180:270:1) -- (0,0); \node(C1) at (-.5,-.5) [circle] {$\widetilde{\mathcal C}_{v_1}$}; \filldraw[draw=none,fill = YellowOrange!50, opacity = .5] (0,0) -- ([shift={(90:1)}]0,0) arc (90:180:1) -- (0,0); \node(C2) at (-.5,.5) [circle] {$\widetilde{\mathcal C}_{v_2}$}; \filldraw[draw=none,fill = YellowOrange!50, opacity = .5] (0,0) -- ([shift={(45:1)}]0,0) arc (45:90:1) -- (0,0); \node(C3) at (.3,.7) [circle] {$\widetilde{\mathcal C}_{v_3}$}; \filldraw[draw=none,fill = YellowOrange!50, opacity = .5] (0,0) -- ([shift={(0:1)}]0,0) arc (0:45:1) -- (0,0); \node(C4) at (.7,.27) [circle] {$\widetilde{\mathcal C}_{v_4}$}; \filldraw[draw=none,fill = YellowOrange!50, opacity = .5] (0,0) -- ([shift={(-90:1)}]0,0) arc (-90:0:1) -- (0,0); \node(C5) at (.5,-.5) [circle] {$\widetilde{\mathcal C}_{v_5}$}; \coordinate (v1) at (-.1,-.1); \coordinate (v1l) at (-1,-.1); \coordinate (v1u) at (-.1,-1); \coordinate (v2) at (-.1,.1); \coordinate (v2l) at (-1,.1); \coordinate (v2o) at (-.1,1); \coordinate (v3) at (.1,0.241421356237); \coordinate (v3o) at (.1,1); \coordinate (v3ro) at (0.63285178118,0.77427313742); \coordinate (v4) at (0.241421356237,.1); \coordinate (v4r) at (1,.1); \coordinate (v4ro) at (0.77427313742,0.63285178118); \coordinate (v5) at (.1,-.1); \coordinate (v5r) at (1,-.1); \coordinate (v5u) at (.1,-1); \filldraw[white] (v1) -- (v1l) -- (v2l) -- (v2) -- (v2o) -- (v3o) -- (v3) -- (v3ro) -- (v4ro) -- (v4) -- (v4r) -- (v5r) -- (v5) -- (v5u) -- (v1u); \draw[YellowOrange] (v1) -- (v1l); \draw[YellowOrange] (v1) -- (v1u); \draw[YellowOrange] (v2) -- (v2l); \draw[YellowOrange] (v2) -- (v2o); \draw[YellowOrange] (v3) -- (v3o); \draw[YellowOrange] (v3) -- (v3ro); \draw[YellowOrange] (v4) -- (v4r); \draw[YellowOrange] (v4) -- (v4ro); \draw[YellowOrange] (v5) -- (v5r); \draw[YellowOrange] (v5) -- (v5u); \draw [decorate,decoration={brace,amplitude=3pt,mirror,raise=2pt},yshift=0pt] (v3o) -- (v2o) node [black,midway,yshift=0.5cm] {$2 R$}; \draw[thick,YellowOrange,->] (0,0) -- (-1.1,0); \draw[thick,YellowOrange] (0,0) -- (0,1); \draw[thick,YellowOrange,->] (0,0) -- (0.70710678118+0.070710678118,0.70710678118+0.070710678118); \draw[thick,YellowOrange,->] (0,0) -- (1.1,0); \draw[thick,YellowOrange,->] (0,0) -- (0,-1.1); \filldraw[YellowOrange] (0,0) circle(1pt); \end{tikzpicture}
    \caption{The normal fan of $\mathcal P$ with the modified cones $\widetilde {\mathcal C}$ indicated.}
    \label{fig:normal_fan_tilde}
    \end{subfigure}%
    \caption{A polytope and its normal fan.}
    \label{fig:normal_fan_poly}
\end{figure}

\begin{lemma}
\label{lmm:cone_limit}
If $\mathcal P \subset \R^n$ is a polytope, then there exists a constant $C > 0$ such that
\begin{gather*} \max_{\bb v \in \mathcal P} \inner{\bb y}{\bb v} - \inner{\bb y}{\bb k} \geq C R \\ \text{ for all } R \geq 0 \text{ and all faces } F \subset \mathcal P \text{ with } \bb k \in \vrt (\mathcal P) \setminus F \\ \text{ and } \bb y \in \mathcal C_{F} \setminus \bigcup_{\dim F' > \dim F} \left( \mathcal C_{F'} + B_R \right), \end{gather*}
where $\vrt(\mathcal P)$ is the set of vertices of $\mathcal P$, the union is over all faces of $\mathcal P$ of higher dimension than $F$ and $B_R \subset \R^n$ is a ball of size $R$.
\end{lemma}
A set of modified cones $\widetilde{\mathcal C}_F = \mathcal C_{F} \setminus \bigcup_{\dim F' > \dim F} \left( \mathcal C_{F'} + B_R \right)$ is depicted in Figure~\ref{fig:normal_fan_tilde}.
Note that the size of the `gaps' between the modified cones is the diameter $2R$ of the ball $B_R$.
\begin{proof}
For a given face $F \subset \mathcal P$ and $\bb k \in \vrt(\mathcal P) \setminus F$, choose some vertex $\bb v \in \vrt(F)$ and consider the hyperplane $H_{\bb v} = \{ \bb y \in \R^n : \inner{\bb y}{(\bb k - \bb v)} = 0 \}$.
By definition of the normal cone, this hyperplane will not intersect with the interior of $\mathcal C_F$ as $\inner{\bb y}{\bb v} \geq \inner{\bb y}{\bb w}$ for all $\bb w \in \mathcal P$ and $\bb y \in \mathcal C_F$.
We can project any point $\bb y \in \mathcal C_F$ onto $H_{\bb v}$ using the orthogonal projection 
$\bb y^\perp = \bb y - (\bb k -\bb v) \frac{\inner{\bb y}{(\bb k -\bb v)}}{\|\bb k -\bb v\|^2} \in H_{\bb v}$. The line segment from $\bb y$ to $\bb y^\perp$ will intersect a face of $\mathcal C_F$.
Let $\mathcal C_{F'}$ be this face. Clearly, $\dim F' > \dim F$.
If $\bb y \not \in \mathcal C_{F'} + B_R$, then the vector $\bb y$ must have a larger distance than $R$ from all points in $\mathcal C_{F'}$.
By construction, the orthogonal projection $\bb y^\perp$ is at least as far away from $\bb y$ as the closest point in $\mathcal C_{F'}$.
Hence, $\| \bb y^\perp -\bb y\| = | \inner{\bb y}{(\bb v - \bb k)} | /\|\bb v - \bb k\| > R$ and it follows that 
\begin{align*} \inner{\bb y}{(\bb v - \bb k)} > R \|\bb v - \bb k\| \text{ for all } \bb y \in \mathcal C_F \setminus (\mathcal C_{F'} + B_R), \bb v \in \vrt(F) \text{ and } \bb k \in \vrt (\mathcal P) \setminus F, \end{align*}
where we used that $\inner{\bb y}{(\bb v - \bb k)} \geq 0$ for all $\bb v\in F$, $\bb k \in \mathcal P$ and $\bb y \in \mathcal C_F$.
To prove the statement, choose $C = \min_{\bb v\neq\bb w \in \vrt(\mathcal P) } \| \bb v - \bb w \|$.
\end{proof}

Theorem~\ref{thm:approx_lower} follows now as a Corollary from the following 
\begin{proposition}
\label{prop:exp_bound_CF}
If $p \in \C[x_1,\ldots,x_N]$ is completely non-vanishing on $\R_{>0}^n$ and $R \geq 0$, then
there is a constant $C > 0$ such that 
\begin{align} \label{eq:exp_bound_CF} C p^\tr(e^{\bb{s}+\bb{t}}) \leq | p(e^{\bb s + \bb t } ) | \text{ for all } \bb s \in \mathcal C_F \text{ and } \bb t \in B_R \text{ for each face $F \subset \NP_p$,} \end{align}
where $B_R\subset \R^n$ is a ball of radius $R$.

\end{proposition}
\begin{proof}
We are going to prove this by induction in the codimension of $F$. Starting with codimension $0$, i.e.\ $F = \NP_p$, we have by Definition~\ref{def:truncated_polynomial}, Lemma~\ref{lmm:truncated_homogenous} and Lemma~\ref{lmm:submult},
\begin{align*} | p(e^{\bb{s} + \bb{t}})| = | p_{\NP_p}(e^{\bb{s} + \bb{t}})| = p^\tr( e^\bb{s}) | p(e^{\bb{t}}) | \geq p^\tr(e^{\bb s + \bb t}) \frac{|p(e^{\bb{t}})| }{p^\tr(e^{\bb t})} \text{ for all } \bb s \in \mathcal C_{\NP_p} \text{ and } \bb t \in \R^n. \end{align*}
Because $p(e^{\bb t})$ is non-vanishing on the compact domain $\bb t \in B_R$, we can choose the constant $C = \min_{\bb t \in B_R} |p(e^{\bb{t}})|/p^\tr(e^{\bb t}) > 0$ to get the desired bound. 

Suppose $F$ is of codimension $d$ and eq.~\eqref{eq:exp_bound_CF} holds for all faces up to codimension $d-1$. 
By the induction hypothesis, for each $R' > 0$ there exists a constant $C >0$ such that eq.~\eqref{eq:exp_bound_CF} is fulfilled in a ball of size $R'$ around all cones $\mathcal C_{F'}$ of lower dimension, $\dim \mathcal C_{F'} < \dim \mathcal C_F$. We therefore only need to prove the existence of such a constant for the smaller domain,
$ \widetilde {\mathcal C}_F = \mathcal C_F \setminus \bigcup_{\dim F' > \dim F} ( \mathcal C_{F'} + B_{R'} ) $, where the union is over all faces of $\NP_p$ of higher dimension than $F$. See Figure~\ref{fig:normal_fan_tilde} for an illustration of this smaller domain.

By definition of the truncated polynomial, we can write $p(e^{\bb{s} + \bb{t}})$ as,
\begin{align*} p(e^{\bb{s} + \bb{t}}) &= p_F(e^{\bb s + \bb{t}}) + \sum_{\bb{k} \in \supp(p) \setminus F} c_{\bb k} e^{\inner{(\bb{s} + \bb{t})}{\bb k}} \text{ for all } \bb s \in \widetilde{\mathcal C}_{F} \text{ and } \bb t \in \R^n, \end{align*}
and estimate using Proposition~\ref{prop:trsupport}, Lemma~\ref{lmm:truncated_homogenous}, Lemma~\ref{lmm:submult} and Lemma~\ref{lmm:cone_limit},
\begin{align*} |p(e^{\bb{s} + \bb{t}})| &\geq p^\tr(e^{\bb s}) \left( |p_F(e^{\bb{t}})| - \sum_{\bb{k} \in \supp(p) \setminus F} |c_{\bb k}| e^{\inner{\bb{t}}{\bb k}} e^{\inner{\bb{s}}{\bb k} - \max_{\bb v \in \NP_p}\inner{\bb s}{\bb v}} \right) \\ &\geq p^\tr(e^{\bb s+ \bb t}) \left( \frac{|p_F(e^{\bb{t}})|}{p^\tr(e^{\bb t})} - e^{-R' C' } \frac{\sum_{\bb{k} \in \supp(p) \setminus F} |c_{\bb k}| e^{\inner{\bb{t}}{\bb k}} }{p^\tr(e^{\bb t})} \right) \\ &\geq p^\tr(e^{\bb s+ \bb t}) \left( \frac{|p_F(e^{\bb{t}})|}{p^\tr(e^{\bb t})} - e^{-R' C' } \sum_{\bb{k} \in \supp(p) \setminus F} |c_{\bb k}| \right) \text{ for all } \bb s \in \widetilde{\mathcal C}_{F} \text{ and } \bb t \in \R^n, \end{align*}
where $C'$ is the constant we obtained from Lemma~\ref{lmm:cone_limit}.
We can choose a constant $C$ such that $2 C = \min_{\bb t \in B_R} |p_F(e^{\bb{t}})|/p^\tr(e^{\bb t}) > 0$ and 
$R' > \frac{1}{C'} \log (\sum_{\bb{k} \in \supp(p) \setminus F} |c_{\bb k}|/C)$. This gives the desired bound.
\end{proof}

\begin{proof}[Proof of Theorem~\ref{thm:approx_lower}] 
Use Proposition~\ref{prop:exp_bound_CF} and the completeness of the normal fan. %
\end{proof}

\begin{remark}
In the proofs of Lemma~\ref{lmm:cone_limit} and Proposition~\ref{prop:exp_bound_CF}, we actually constructed explicit bounds for the constants in Theorem~\ref{thm:approx} which depend on the geometry of the relevant polytopes and polynomials. These explicit bounds might be useful for further considerations, but we will not make use of them in this article.
\end{remark}

\section{Geometric sector decomposition}
\label{sec:tropsecdec}

The rough overall plan of our take on the integral in eq.~\eqref{eq:integral} is as follows: we can trivially `factorize' its integrand and write it as
\begin{align*} I = \int_{\P_{>0}^{n-1}} \frac{\prod_{i} a_{i}^\tr(\bb{x})^{\Re \nu_i}}{\prod_{j} b_{j}^\tr(\bb{x})^{\Re \rho_j}} \left( \frac{\prod_{i} a_{i}(\bb x)^{\nu_i}/a_{i}^\tr(\bb x)^{\Re \nu_i} }{\prod_{j} b_{j}(\bb x)^{\rho_j}/b_{j}^\tr(\bb x)^{\Re \rho_j}} \right) \Omega . \end{align*}
The second factor is bounded by Corollary~\ref{crll:homo_approx} as long as the $\{b_j\}$ polynomials are completely non-vanishing. The first term has a geometric interpretation in terms of the polytopes 
$\mathcal{A} = \sum_{i} (\Re \nu_i) \NP_{a_i}$ and $\mathcal{B} = \sum_{j} (\Re \rho_j) \NP_{b_j}$ as defined in Theorem~\ref{thm:convergence}.

As before it will be handy to change to logarithmic coordinates to expose this geometric interpretation. The component-wise exponential $\Exp: \R^{n} \rightarrow \R^{n}_{>0}, \bb y \mapsto e^{\bb y} = (e^{y_1}, \ldots, e^{y_n})$, extends to a smooth bijective map $\Exp: \R^n/\one \R \rightarrow \P^{n-1}_{>0}$, as $\Exp$ respects the respective equivalence relation. That means if $\bb x = e^{\bb y}$ and $\bb x' = e^{\bb y'}$ with $\bb y, \bb y' \in \R^n$,
then $\bb y' = \bb y + \mu \one \text{ for some } \mu \in \R$ if and only if
$\bb x' = \lambda \bb x \text{ for some }\lambda \in \R_{>0}. $
For this reason the quotient $\R^n/\one \R$ is also called \emph{tropical projective space}.

By Proposition~\ref{prop:trsupport} and the definition of the weighted Minkowski sum with $\bb x = e^{\bb y}$
\begin{align} \begin{aligned} \label{eq:atrbtr_exp} \frac{\prod_{i} a_{i}^\tr(\bb x)^{\Re \nu_i}}{\prod_{j} b_{j}^\tr(\bb x)^{\Re \rho_j}} &= \exp \left( \sum_{i} \Re \nu_i \max_{\bb v\in \NP_{a_i}}\inner{\bb y}{\bb v} - \sum_{j} \Re \rho_j \max_{\bb v\in \NP_{b_j}}\inner{\bb y}{\bb v} \right) \\ &= \exp\left(\max_{\bb v\in \mathcal A}\inner{\bb y}{\bb v} - \max_{\bb v \in \mathcal B}\inner{\bb y}{\bb v} \right). \end{aligned} \end{align}
If $\mathcal A$ and $\mathcal B$ fulfill the
requirements \normalfont{R1} and \normalfont{R2} of Theorem~\ref{thm:convergence}, this exponent is falling sufficiently fast for large $\bb y$ for the integral in eq.~\eqref{eq:integral} to be convergent.

\begin{lemma}
\label{lmm:ABsupportlimit}
Let $\mathcal A$ and $\mathcal B$ be the polytopes defined in Theorem~\ref{thm:convergence}. If $\mathcal A$ and $\mathcal B$ fulfill the
requirements \normalfont{R1} and \normalfont{R2} of Theorem~\ref{thm:convergence},
then there is a constant $\varepsilon >0$ such that 
\begin{align*} \max_{\bb v \in \mathcal B} \inner{\bb y}{\bb v} - \max_{\bb v \in \mathcal A} \inner{\bb y}{\bb v} \geq \varepsilon \|\bb y\|_{\R^n/\one \R} \text{ for all } \bb y \in \R^n/\one \R, \end{align*}
where $\| \bb y \|_{\R^n/\one \R}=\inf_{\mu \in \R} \| \bb y+\mu \one\|$ is the norm on the quotient space $\R^n/\one \R$ induced from the standard norm $\|\cdot \|$ on $\R^n$.
\end{lemma}
\begin{proof}
First note that the inequality in the statement is well-defined for $\bb y \in \R^n/\one \R$, as $\mathcal A$ and $\mathcal B$ both lie in the same hyperplane $\mathcal A, \mathcal B \subset H_\xi = \{\bb{v} \in \R^n : \inner{\one}{\bb{v}} = \xi\}$ for some $\xi \in \R$. Therefore, $\inner{\bb y}{\bb v} = \inner{\bb y}{\bb w}$ for all $\bb v \in \mathcal A$, $\bb w \in \mathcal B$ and $\bb y \in \one \R$.

As $\mathcal B$ is full-dimensional in $H_\xi$ (see Remark~\ref{rmk:homogeneous}) 
and $\mathcal A \subset \relint \mathcal B$,
we can Minkowski add a ball $B_\varepsilon$ to $\mathcal A$ such that $\mathcal A + B_\varepsilon \subset \mathcal B$, provided that this ball only extends in the subspace orthogonal to $\one \R$ which is parallel to $H_\xi$ and $\varepsilon$ is sufficiently small. Let $B_\varepsilon = \{ \bb v \in \R^n : \inner{\one}{\bb v} = 0 \text{ and } \|\bb v\| \leq \varepsilon \}$ be such a ball.  The resulting convex set $\mathcal A + B_\varepsilon$ is the \emph{outer parallel body} of $\mathcal A$ restricted to its affine hull.

Observe that $\max_{\bb v \in B_\varepsilon} \inner{\bb y}{\bb v} = \varepsilon \inf_{\mu \in \R} \| \bb y+\mu \one\|= \varepsilon \| \bb y \|_{\R^n/\one \R}$. By the definition of the Minkowski sum and because $\mathcal A + B_\varepsilon \subset \mathcal B$, 
\begin{gather*} \max_{\bb v \in \mathcal A} \inner{\bb y}{\bb v} + \varepsilon \| \bb y \|_{\R^n/\one \R} = \max_{\bb v \in \mathcal A} \inner{\bb y}{\bb v} + \max_{\bb v \in B_\varepsilon} \inner{\bb y}{\bb v} \\ = \max_{\bb v \in \mathcal A + B_\varepsilon} \inner{\bb y}{\bb v} \leq \max_{\bb v \in \mathcal B} \inner{\bb y}{\bb v} \text{ for all } \bb y \in \R^n. \qedhere \end{gather*}
\end{proof}

To be able to handle Kaneko-Ueda geometric sector decomposition with our tropical approach, we will need additional tools from convex geometry. 

A cone is called \emph{pointed} if it contains no $1$-dimensional subspace. A fan is pointed if all its cones are pointed. If a polytope $\mathcal P \subset \R^n$ is full-dimensional, its normal fan $\mathcal F$ is pointed. For lower dimensional polytopes $\mathcal P \subset \R^n$, the normal cone $\mathcal C_{\mathcal P}$ associated to the polytope itself is non-trivial: $\mathcal C_{\mathcal P} = {\mathcal P}^\perp = \left \{ \bb y \in \R^n : \inner{\bb y}{\bb v} = \inner{\bb y}{\bb w} \text{ for all } \bb v, \bb w \in \mathcal P \right\}$. It consists of all linear functionals that are constant on $\mathcal P$. This subspace is contained in each cone of the normal fan. Taking the quotient with respect to this subspace within each cone in the normal fan $\mathcal C \in \mathcal F$ results in a pointed fan $\mathcal F/\mathcal P^\perp$ on the quotient vector space $\R^n/\mathcal P^\perp$. This fan is the \emph{reduced normal fan}.

Given a fan $\mathcal F$, another fan $\mathcal F'$ \emph{refines} $\mathcal F$ if every cone in $\mathcal F$ is a union of cones in $\mathcal F'$.  If $\mathcal F$ and $\mathcal G$ are both fans, then their \emph{common refinement} is defined as $\mathcal F \wedge \mathcal G = \{ \mathcal C \cap \mathcal C' : \mathcal C \in \mathcal F, \mathcal C' \in \mathcal G \}$.  
Let $\mathcal F_{\mathcal{AB}}$ be the common refinement of the normal fans of the polytopes $\mathcal A$ and $\mathcal B$ which were defined in Theorem~\ref{thm:convergence}.
Recall that the polytopes $\mathcal{A}, \mathcal{B}$ are not full-dimensional, because they are weighted Minkowski sums of the Newton polytopes of homogeneous polynomials (see Remark~\ref{rmk:homogeneous}). As $\mathcal B$ is required to be full-dimensional in a ($n-1$)-dimensional hyperplane which is orthogonal to the $\one$-vector, we have $\mathcal B^\perp = \one \R$. We will therefore consider the reduced refined normal fan $\mathcal F_{\mathcal{AB}} / \one \R$, which is pointed.
The following lemma identifies the exponentiated cones of $\mathcal C \in \mathcal F_{\mathcal{AB}} / \one \R$ as the domains where the function $\prod_{i} a_{i}^\tr(\bb x)^{\Re \nu_i}/\prod_{j} b_{j}^\tr(\bb x)^{\Re \rho_j}$ behaves like a monomial.

\begin{lemma}
\label{lmm:abtrop_w}
Let $\mathcal A$ and $\mathcal B$ be the polytopes defined in Theorem~\ref{thm:convergence}. If $\mathcal A$ and $\mathcal B$ fulfill the
requirements \normalfont{R1} and \normalfont{R2} of Theorem~\ref{thm:convergence}
and $\mathcal C$ is a cone in the reduced common refinement $\mathcal F_{\mathcal{AB}}/\one \R$, then
\begin{align*} \frac{\prod_{i} a_{i}^\tr(\bb x)^{\Re \nu_i}}{\prod_{j} b_{j}^\tr(\bb x)^{\Re \rho_j}} = \bb x^{-\bb w} \text{ for all } \bb x \in \Exp (\mathcal C), \end{align*}
where 
$\bb w = \bb w_{\mathcal B} - \bb w_{\mathcal A}$ 
and $\bb w_{\mathcal A} \in \mathcal A, \bb w_{\mathcal B} \in \mathcal B$ 
such that 
$\inner{\bb y}{\bb w_{\mathcal A}} = \max_{\bb v \in \mathcal A} \inner{\bb y}{\bb v}$
and $\inner{\bb y}{\bb w_{\mathcal B}} = \max_{\bb v \in \mathcal B} \inner{\bb y}{\bb v}$
for all $\bb y \in \mathcal C$. Moreover, $\inner{\one}{\bb w} = 0$ and $\inner{\bb y}{\bb w} > 0$ for all $\bb y \in \mathcal C \setminus \{0\}$.
\end{lemma}
\begin{proof}
As $\mathcal C$ is a refinement of the normal fans of $\mathcal A$ and $\mathcal B$, there must be normal cones $\mathcal C_{F_{\mathcal A}}^{\mathcal A}$ and $\mathcal C_{F_{\mathcal B}}^{\mathcal B}$ associated to respective faces $F_{\mathcal A} \subset \mathcal A$ and $F_{\mathcal B} \subset \mathcal B$ such that $\mathcal C \subset \mathcal C_{F_{\mathcal A}}^{\mathcal A} \cap \mathcal C_{F_{\mathcal B}}^{\mathcal B}$.
Hence,
$ \max_{\bb v\in \mathcal B}\inner{\bb y}{\bb v} - \max_{\bb v \in \mathcal A}\inner{\bb y}{\bb v} = \inner{\bb y}{(\bb w_{\mathcal B} - \bb w_{\mathcal A} )} \text{ for all } \bb y \in \mathcal C$,
where we can choose arbitrary $\bb w_{\mathcal A} \in F_{\mathcal A}$ and $\bb w_{\mathcal B} \in F_{\mathcal B}$ by definition of the normal cone in eq.~\eqref{eq:normal_cone}.

Since $\mathcal A$ and $\mathcal B$ are required to lie in the same hyperplane orthogonal to the $\one$-vector, we also have $\inner{\one}{ (\bb w_{\mathcal B} - \bb w_{\mathcal A} ) } = \inner{\one}{\bb w} = 0$.
Due to Lemma~\ref{lmm:ABsupportlimit}, 
$\max_{\bb v\in \mathcal B}\inner{\bb y}{\bb v} - \max_{\bb v \in \mathcal A}\inner{\bb y}{\bb v} = \inner{\bb y}{(\bb w_{\mathcal B} - \bb w_{\mathcal A} )} \geq \varepsilon \| \bb y \|_{\R^n/\one \R }> 0 \text{ for all } \bb y \in \mathcal C \setminus \{0\}$ and the statement follows from eq.~\eqref{eq:atrbtr_exp}.
\end{proof}

A cone $\mathcal C$ is \emph{simplicial} if it is generated as, $\mathcal C = \{ \sum_{k=1}^d \lambda_k \bb u^{(k)} : \lambda_k \geq 0 \}$ where $\bb u^{(1)}, \ldots \bb u^{(d)}$ are linear independent.  
For a given cone $\mathcal C$, we can always find a set of simplicial cones $\mathcal C^\Delta_1, \mathcal C^\Delta_2, \ldots$ such that $\mathcal C = \biguplus_{i} \mathcal C^\Delta_i$ and $\mathcal C^\Delta_i \cap \mathcal C^\Delta_j$ is a simplicial cone which is a face of both $\mathcal C^\Delta_i$ and $\mathcal C^\Delta_j$. Such a set of simplicial cones is called a \emph{triangulation} of $\mathcal C$. %
Let $\mathcal F_{\mathcal{AB}}^\Delta/\one \R$ be a \emph{simplicial refinement} of $\mathcal F_{\mathcal{AB}}/\one \R$, i.e.\ a refinement of $\mathcal F_{\mathcal{AB}}/\one \R$ such that each cone in $\mathcal F_{\mathcal{AB}}^\Delta/\one \R$ is simplicial. 

A feature of simplicial cones is that there are convenient coordinates describing points in their interior. This fact is important while proving the following lemma:
\begin{lemma}
\label{lmm:cone_integral}
If a pointed simplicial cone $\mathcal C \subset \R^n/\one \R$ is generated by
linear independent vectors $\bb u^{(1)}, \ldots, \bb u^{(n-1)} \in \R^n/\one \R$, i.e.\ $\mathcal C = \{ \sum_{k=1}^{n-1} \lambda_k \bb u^{(k)}: \lambda_k \geq 0 \}$, $\bb w \in \R^n$ with $\inner{\one}{\bb w} = 0$ and
$\inner{\bb y}{\bb w} > 0$ for all $\bb y \in \mathcal C\setminus\{0\}$, then 
\begin{align*} \label{eq:cone_integral} \int_{\Exp(\mathcal C)} \bb x^{-\bb w} f(\bb x) \Omega = \frac{\left|\det(\bb u^{(1)}, \ldots, \bb u^{(n-1)},\one) \right|}{ \prod_{k=1}^{n-1} \inner{\bb u^{(k)}}{\bb w} } \int_{[0,1]^{n-1}} f\left( \bb x(\bb \xi) \right) \prod_{k=1}^{n-1} \dd \xi_k, \end{align*}
where 
$f:\P_{>0}^{n-1} \rightarrow \C$ is a measurable homogeneous function of degree $0$ and
$\bb x(\bb \xi) \in \Exp(\mathcal C)$ is given component-wise by $x_k = \prod_{i = 1}^{n-1} \xi_i^{-u_k^{(i)}/(\inner{\bb u^{(i)}}{\bb w})}$. %
\end{lemma}

\begin{remark}
By slightly abusing the notation, we identified the vectors $\bb u^{(1)}, \ldots, \bb u^{(n-1)}$ with appropriate representatives in $\R^n$ in the statement of this lemma.  The value of the integral does not depend on the specific choice of representatives, because $\inner{\one}{\bb w} = 0$ and $f$ is homogeneous of degree zero. Hence, the expression on the right hand side is invariant under shifts $\bb u^{(k)} \rightarrow \bb u^{(k)} + \mu_k \one$ for all $\mu_k \in \R$.
It is also invariant under rescalings of the vectors $\bb u^{(k)} \rightarrow \lambda_k \bb u^{(k)}$ for all $\lambda_k > 0$ as it should be due to the equivalence of the cone representation. 
Even though changing the representatives of the $\bb u$-vectors modifies the vector $\bb x(\bb \xi)$, it only does so by an overall scaling, which does not modify the point in $\P^{n-1}_{>0}$, which $\bb x(\bb \xi)$ represents.
\end{remark}

\begin{proof}
Start by changing to logarithmic coordinates $\bb x = e^{\bb y}$,
\begin{align*} \int_{\Exp(\mathcal C)} \bb x^{-\bb w} f(\bb x) \Omega = \int_{\mathcal C} e^{-\inner{\bb y}{\bb w}} f(e^{\bb y}) \widetilde \Omega \end{align*}
where
$\widetilde{\Omega} = \Exp^* \Omega= \sum_{k=1}^n (-1)^{n-k}\dd y_1\wedge \ldots \wedge \widehat {\dd y_k } \wedge \ldots \wedge \dd y_n$ is the pullback of $\Omega$ under $\Exp$.
Using \emph{barycentric coordinates} $\bb y = \sum_{k=1}^{n-1} \bb u^{(k)} \lambda_k$ shows that this is equal to
\begin{gather*}   = |\det(\bb u^{(1)}, \ldots, \bb u^{(n-1)},\one)| \int_{\R_{> 0}^{n-1} } e^{-\sum_{k=1}^{n-1} \lambda_k \inner{\bb u^{(k)}}{\bb w}} f\left(e^{\sum_{k=1}^{n-1} \lambda_k \bb u^{(k)}}\right) \prod_{k=1}^{n-1} \dd \lambda_k. \end{gather*}
The form of the determinant follows from the form $\widetilde \Omega$, Laplace's expansion and 
\begin{align*} \dd y_1 \wedge \ldots \wedge \widehat{\dd y_k} \wedge \ldots \wedge \dd y_n = \begin{vmatrix} u_1^{(\mathcal C,1)} & \hdots & \widehat{ u_k^{(\mathcal C,1)}} & \hdots & u_n^{(\mathcal C,1)}\\ \vdots & \ddots& \vdots & \ddots& \\ u_1^{(\mathcal C,n-1)} & \hdots & \widehat{ u_k^{(\mathcal C,n-1)}} & \hdots & u_n^{(\mathcal C,n-1)} \end{vmatrix} \dd \lambda_1 \wedge \ldots \wedge \dd \lambda_{n-1}. \end{align*}
As $\inner{\bb y}{\bb w} > 0$ for all $\bb y \in \mathcal C\setminus\{0\}$, it follows that $\inner{\bb u^{(k)}}{\bb w} > 0$ for all $k \in \{ 1,\ldots, n-1\}$.
We can therefore change variables via $\lambda_k = -\frac{1}{\inner{\bb u^{(k)}}{\bb w}} \log \xi_k$ which proves the statement.
\end{proof}

With these tools at hand, we are ready to give our tropical formulation of geometric sector decomposition:

\begin{theorem}[Geometric sector decomposition]
\label{thm:secdec}
Let $\mathcal A$ and $\mathcal B$ be the polytopes defined in Theorem~\ref{thm:convergence}. If $\mathcal A$ and $\mathcal B$ fulfill the
requirements \normalfont{R1} and \normalfont{R2} of Theorem~\ref{thm:convergence}
and $\mathcal M_{\mathcal{AB}}^\Delta \subset \mathcal F_{\mathcal{AB}}^\Delta/\one \R$ is the set of maximal cones, i.e.\ the cones of maximal dimension, in a simplicial refinement of the reduced common normal fan of $\mathcal A$ and $\mathcal B$, then we can write the integral 
\begin{gather*}  I[f] = \int_{\P_{>0}^{n-1}} \frac{\prod_{i} a_{i}^\tr(\bb x)^{\Re \nu_i}}{\prod_{j} b_{j}^\tr(\bb{x})^{\Re \rho_j}} f(\bb x) \Omega.      \end{gather*}
as a sum 
$I[f] = \sum_{\mathcal C \in \mathcal M_{\mathcal{AB}}^\Delta} I_{\mathcal C}[f]$ with
\begin{align} \label{eq:evalCintegral} I_{\mathcal C}[f] &= \frac{ \left|\det(\bb u^{(\mathcal C,1)}, \ldots, \bb u^{(\mathcal C,n-1)},\one) \right| }{ \prod_{k=1}^{n-1} \inner{\bb u^{(\mathcal C,k)}}{\bb w^{(\mathcal C)}} } \int_{[0,1]^{n-1}} f \left( \bb x^{(\mathcal C)}(\bb \xi) \right) \prod_{k=1}^{n-1} \dd \xi_k, \end{align}
where 
\begin{itemize}
\item
$f:\P_{>0}^{n-1} \rightarrow \C$ is a measurable homogeneous function of degree $0$,
\item
$\bb w^{(\mathcal C)} = \bb w^{(\mathcal C)}_{\mathcal B} - \bb w^{(\mathcal C)}_{\mathcal A}$ with
some $\bb w^{(\mathcal C)}_{\mathcal A} \in \mathcal A$ and $\bb w^{(\mathcal C)}_{\mathcal B} \in \mathcal B$ such that
$\inner{\bb y}{\bb w^{(\mathcal C)}_{\mathcal A} } = \max_{\bb v \in \mathcal{A}} \inner{\bb y}{\bb v}$ and 
$\inner{\bb y}{\bb w^{(\mathcal C)}_{\mathcal B} } = \max_{\bb v \in \mathcal{B}} \inner{\bb y}{\bb v}$ for all $\bb y \in \mathcal C$,
\item
the vectors $\bb u^{(\mathcal C,1)}, \ldots, \bb u^{(\mathcal C,n-1)} \in \R^n/\one \R$ span the simplicial cone $\mathcal C$ such that $\mathcal C = \{ \sum_{k=1}^{n-1} \lambda_k \bb u^{(\mathcal C,k)}: \lambda_k \geq 0 \} \subset \R^n/\one \R$,
\item
$\bb x^{(\mathcal C)}(\bb \xi) \in \Exp(\mathcal C)$ is given component-wise by $x_k^{(\mathcal C)} = \prod_{i = 1}^{n-1} \xi_i^{-u_k^{(\mathcal C,i)}/(\inner{\bb u^{(\mathcal C,i)}}{\bb w^{(\mathcal C)}})}$ and
\item
the prefactor $ \frac{ \left|\det(\bb u^{(\mathcal C,1)}, \ldots, \bb u^{(\mathcal C,n-1)},\one) \right| }{ \prod_{k=1}^{n-1} \inner{\bb u^{(\mathcal C,k)}}{\bb w^{(\mathcal C)}} }$ is finite and positive for each $\mathcal C \in \mathcal M_{\mathcal{AB}}^\Delta$.
\end{itemize}
\end{theorem}

\begin{proof}
The fan $\mathcal F_{\mathcal{AB}}^\Delta/\one \R$ is complete, i.e.\ it corresponds to a partition of $\R^{n}/\one \R = \biguplus_{\mathcal C \in F_{\mathcal{AB}}^\Delta/\one \R} \mathcal C$. 
Because $\Exp: \R^n/\one \R \rightarrow \P^{n-1}_{>0}$ is smooth and bijective this partition gives also a partition of $\P^{n-1}_{>0}= \biguplus_{\mathcal C \in F_{\mathcal{AB}}^\Delta/\one \R} \Exp (\mathcal C)$. 
Since we would like to integrate over $\P^{n-1}_{>0}$ or (equivalently over $\R^n/\one \R$) it is enough to only consider the cones of maximal dimension $\mathcal M_{\mathcal{AB}}^\Delta \subset \mathcal F_{\mathcal{AB}}^\Delta/\one \R$ as other cones in $\mathcal F_{\mathcal{AB}}^\Delta/\one \R$ only describe measure zero subsets of $\P^{n-1}_{>0}$. Hence,
\begin{align*}  I &= \sum_{\mathcal C \in \mathcal M_{\mathcal{AB}}^\Delta} I_{\mathcal C} & I_{\mathcal C} &= \int_{\Exp(\mathcal C)} \frac{\prod_{i} a_{i}^\tr(\bb x)^{\Re \nu_i}}{\prod_{j} b_{j}^\tr(\bb{x})^{\Re \rho_j}} f(\bb x) \Omega. \end{align*}
Because each cone $\mathcal C \in \mathcal M_{\mathcal{AB}}^\Delta$ refines a cone in $\mathcal F_{\mathcal{AB}}/\one \R$ and because of Lemma~\ref{lmm:abtrop_w}, 
\begin{align*}  I_{\mathcal C} &= \int_{\Exp(\mathcal C)} \bb x^{-\bb w^{(\mathcal C)}} f(\bb x) \Omega \text{ for all } \mathcal C \in \mathcal M_{\mathcal{AB}}^\Delta, \end{align*}
where 
$\bb w^{(\mathcal C)} = \bb w^{(\mathcal C)}_{\mathcal B} - \bb w^{(\mathcal C)}_{\mathcal A}$ with
some $\bb w^{(\mathcal C)}_{\mathcal A} \in \mathcal A$ and $\bb w^{(\mathcal C)}_{\mathcal B} \in \mathcal B$ such that
$\inner{\bb y}{\bb w^{(\mathcal C)}_{\mathcal A} } = \max_{\bb v \in \mathcal{A}} \inner{\bb y}{\bb v}$ and 
$\inner{\bb y}{\bb w^{(\mathcal C)}_{\mathcal B} } = \max_{\bb v \in \mathcal{B}} \inner{\bb y}{\bb v}$ for all $\bb y \in \mathcal C$. By Lemma~\ref{lmm:abtrop_w}, we also have $\inner{\one}{\bb w^{(\mathcal C)}} = 0$ and $\inner{\bb y}{\bb w^{(\mathcal C)}} > 0$ for all $\bb y \in \mathcal C \setminus \{0\}$.

Eq.~\eqref{eq:evalCintegral} follows from Lemma~\ref{lmm:cone_integral}, because  we can always pick a set of generators $\bb u^{(\mathcal C,1)}, \ldots, \bb u^{(\mathcal C,n-1)} \in \R^n/\one \R$ for every simplicial cone $\mathcal C \in \mathcal M_{\mathcal{AB}}^\Delta$.
As $\inner{\bb y}{\bb w^{(\mathcal C)}} > 0$ for all $\bb y \in \mathcal C \setminus \{0\}$, we also have $\inner{\bb u^{(\mathcal C,k)}}{\bb w^{(\mathcal C)}} > 0$ for all $k \in \{ 1,\ldots,n-1\}$. The positivity of the determinant is obvious because of the linear independence of the vectors $\bb u^{(\mathcal C,k)}$.
\end{proof}

If we specify $f(\bb x) = R_{a/b}(\bb x)$ given by 
\begin{align} \label{eq:def_residualab} R_{a/b}(\bb x) = \frac{\prod_{i} a_{i}(\bb x)^{\nu_i}/a_{i}^\tr(\bb x)^{\Re \nu_i} }{\prod_{j} b_{j}(\bb x)^{\rho_j}/b_{j}^\tr(\bb x)^{\Re \rho_j}}, \end{align}
in Theorem~\ref{thm:secdec}, we recover the integral in eq.~\eqref{eq:integral}.

\begin{proof}[Proof of Theorem~\ref{thm:convergence}]
We only need to prove that each sector integral in the geometric sector decomposition of Theorem~\ref{thm:secdec} with $f(\bb x) = R_{a/b}(\bb x)$ from eq.~\eqref{eq:def_residualab} is finite. As all the denominator polynomials $\{b_j\}$ are completely non-vanishing, Corollary~\ref{crll:homo_approx} implies that $ |R_{a/b}(\bb x)|$ is bounded on $\P_{>0}^{n-1}$. Hence, each integral $I_{\mathcal C}[R_{a/b}]$ is finite.
\end{proof}

Theorem~\ref{thm:secdec} provides a sector decomposition as it was formulated in eq.~\eqref{eq:secdec_integral}, because the sector integrands in Theorem~\ref{thm:secdec} are bounded as long as the function $f$ is bounded on $\P_{>0}^{n-1}$. This way, Theorem~\ref{thm:secdec} not only ensures finiteness of the integral in eq.~\eqref{eq:integral} under appropriate conditions, but also allows to evaluate the integral via Monte Carlo quadrature. 

If we have triangulated the reduced normal fan $\mathcal F_{\mathcal{AB}}/\one \R$, i.e.~we have computed a simplicial refinement $\mathcal F^\Delta_{\mathcal{AB}}/\one \R$ and stored the vectors $\bb u^{(\mathcal C,1)}, \ldots, \bb u^{(\mathcal C,n-1)}$ and $\bb w^{(\mathcal C)}$ for each maximal cone $\mathcal C \in \mathcal M^\Delta_{\mathcal{AB}} \subset \mathcal F^\Delta_{\mathcal{AB}}/\one \R$ in a table, then we can estimate the integral using Algorithm~\ref{alg:basic}.
\begin{algorithm}[H]
\begin{algorithmic}[0]
\ForAll{maximal cones $\mathcal C \in \mathcal M_{\mathcal{AB}}^\Delta$}%
\For{ $\ell \in 1,\ldots,N$  } 
\State{Draw a random vector $\bb \xi \in [0,1]^{n-1}$ from the distribution $1=\int_{[0,1]^{n-1}} \prod_{i=1}^{n-1} d\xi_i$.  }
\State{Set $x_k^{(\ell)} = \prod_{i = 1}^{n-1} \xi_i^{-u_k^{(\mathcal C,i)}/\inner{\bb u^{(\mathcal C,i)}}{\bb w^{(\mathcal C)}}}$ for all $k=1,\ldots,n$. }
\EndFor
\State{%
Set
$ I_{\mathcal C}^{(N)}[R_{a/b}] = \frac{1}{N} \frac{ \left| \det(\bb u^{({\mathcal C},1)}, \ldots, \bb u^{(\mathcal C, n-1)},\one ) \right| }{ \prod_{k=1}^{n-1} \inner{\bb u^{(\mathcal C,k)}}{\bb w^{(\mathcal C)}} } \sum_{\ell=1}^N R_{a/b}(\bb x^{(\ell)}). $}
\EndFor
\State{Return $I^{(N)} = \sum_{{\mathcal C} \in \mathcal M_{\mathcal{AB}}^\Delta}I_{\mathcal C}^{(N)}[R_{a/b}]$.}
\end{algorithmic}
\caption{Basic Monte Carlo quadrature of Euler-Mellin integrals}
\label{alg:basic}
\end{algorithm}
\begin{proposition}
\label{prop:alg_basic}
If the conditions of Theorem~\ref{thm:convergence} are fulfilled, then the random value $I^{(N)}$ returned by Algorithm~\ref{alg:basic} has expectation value equal to the integral in eq.~\eqref{eq:integral}, $I = \E[I^{(N)}]$ and $\var [I^{(N)}] = \frac{C}{N}$ with some constant $C \geq 0$.
\end{proposition}
\begin{proof}
Algorithm~\ref{alg:basic} is an application of Theorem~\ref{thm:montecarlo} on the integral $I_\mathcal{C}[f]$ for each cone $\mathcal C \in \mathcal M_{\mathcal{AB}}^\Delta$ from Theorem~\ref{thm:secdec} with $f(\bb x) = R_{a/b}(\bb x)$: 
\begin{align*} \int_{[0,1]^{n-1}} R_{a/b}\left( \bb x^{(\mathcal C)}(\bb \xi) \right) \prod_{k=1}^{n-1} \dd \xi_k. \end{align*}
As $|R_{a/b}(\bb x)|$ is bounded on $\P_{>0}^{n-1}$ and $\bb x^{(\mathcal C)}(\bb \xi) \in \P_{>0}^{n-1}$ by construction, the integrand is bounded and therefore also square integrable. Hence, there is a constant $C_{\mathcal C}\geq 0$ for each cone integral such that $\var[I_\mathcal{C}^{(N)}] = C_{\mathcal C}/N$ and $\var [I^{(N)}] = \sum_{{\mathcal C} \in \mathcal M_{\mathcal{AB}}^\Delta} \var [ I_{\mathcal C}^{(N)}] = C/N$.
\end{proof}

Effectively, Proposition~\ref{prop:alg_basic} ensures that we can consider the random variable $I^{(N)}$ as an approximation for $I$ with relative accuracy $\delta = \frac{1}{I} \sqrt{C/N}$. Estimating the constant $C$ is usually easy in practice: as long as sufficiently high powers of the integrand $f(\bb x)$ are integrable, we can also use Theorem~\ref{thm:montecarlo} to estimate $\var[ f(\bb x)]$.

Variants of Algorithm~\ref{alg:basic} are implemented e.g.\ as \texttt{SecDec-3}~\cite{Borowka:2015mxa} and as \texttt{FIESTA~3}~\cite{Smirnov:2013eza}. 
Both these implementations provide a variety of different ways to perform the preprocessing triangulation step which computes $\mathcal M_{\mathcal{AB}}^\Delta$. A dedicated tool to perform such a triangulation is \texttt{Normaliz} \cite{bruns2010normaliz} which is also used internally in \texttt{SecDec-3}. %
Subsequently, both programs use a version of the VEGAS algorithm \cite{Lepage:1977sw,Hahn:2004fe} to numerically integrate $I_{\mathcal C}[f]$ in eq.~\eqref{eq:evalCintegral} over the unit hypercube $[0,1]^{n-1}$ or, equivalently, to execute the inner loop of Algorithm~\ref{alg:basic} for each individual cone $\mathcal C \in \mathcal M_{\mathcal{AB}}^\Delta$. 

We can estimate the computational complexity of the algorithm by counting the number of necessary evaluations of the function $R_{a/b}(\bb x)$. This is justified because we can assume that the runtime to evaluate $R_{a/b}(\bb x)$ overshadows the time it takes to compute a random vector $\bb \xi \in [0,1]^{n-1}$ and the value of $\bb x(\bb \xi) \in\P^{n-1}_{>0}$ from it. 
Therefore, the estimation step summarized in Algorithm~\ref{alg:basic} needs $N |\mathcal M_{\mathcal{AB}}^\Delta|$ evaluations to produce the estimate $I^{(N)}$ for the integral in eq.~\eqref{eq:integral}. Equivalently, as the relative accuracy $\delta \approx I^{(N)}/I$ of the resulting estimate is inverse proportional to $\sqrt{N}$, the number of evaluations needed is proportional to $\delta^{-2} |\mathcal M_{\mathcal{AB}}^\Delta|$ to achieve an estimate of $\delta$ accuracy.

A severe bottleneck is the number of maximal cones $|\mathcal M_{\mathcal{AB}}^\Delta|$ which tends to grow exponentially with growing dimension $n$ of the problem. A particularly unsatisfying aspect of this bottleneck is that the value of the individual sector contributions $I_{\mathcal C}$ typically varies quite much in magnitude. Consequently, only a fraction of the geometric sector contributions in eq.~\eqref{eq:evalCintegral} are relevant for the overall integral $I$ and much of the computational effort spent to estimate each of the integrals $I_{\mathcal C}$ is wasted. In the next section, we will explain how to overcome this bottleneck.

\section{Tropical sampling}
\label{sec:trop_sampling}
In summary, the strategy to overcome this problem is the following: instead of numerically integrating each of the sector integrals individually and eventually summing all the resulting numbers to obtain an estimate for the integral in eq.~\eqref{eq:integral}, we can use a more `inclusive' Monte Carlo approach, where we evaluate both the individual integrals $I_{\mathcal C}[R_{a/b}]$ in eq.~\eqref{eq:evalCintegral} and \emph{the sum over these integrals} $\sum_{\mathcal C \in \mathcal M_{\mathcal{AB}}^\Delta} I_{\mathcal C}[R_{a/b}]$ via Monte Carlo methods. This approach is much more efficient than the traditional one because there is a canonical way to perform \emph{importance sampling} on the sum. That means that we can expose the individual sectors to our sampler `undemocratically' such that more important sectors are sampled more often than less important contributions.

To do this it is convenient to define a tropically approximated version of the integral in eq.~\eqref{eq:integral}:
\begin{align} \label{eq:integral_trop} I^\tr = \int_{\P_{>0}^{n-1}} \frac{\prod_{i} a^\tr_{i}(\bb{x})^{\Re \nu_i}}{\prod_{j} b^\tr_{j}(\bb{x})^{\Re \rho_j}} \Omega. \end{align}
Such tropically approximated integrals have been considered as a simple avatar of \emph{period Feynman integrals} \cite{Panzer:2019yxl} and identified to appear in the weak string coupling limit \cite{Arkani-Hamed:2019mrd}. Moreover, this tropically approximated integral also gives rise to the \emph{canonical function} of a polytope under certain conditions on the polynomials $\{a_i\}$ and $\{b_j\}$, which has applications in the theory of scattering amplitudes \cite{Arkani-Hamed:2017tmz,Arkani-Hamed:2019mrd}.

It follows from Theorem~\ref{thm:secdec} with $f(\bb x) =1$ that $I^\tr$ is finite and that $I^\tr > 0$, provided that the conditions \normalfont{R1} and \normalfont{R2} on the polytopes $\mathcal A$ and $\mathcal B$ in Theorem~\ref{thm:convergence} are fulfilled. Moreover, the integrand in eq.~\eqref{eq:integral_trop} is obviously positive for all $\bb x \in \P^{n-1}_{>0}$. Hence, we can define a probability distribution given by the differential form,
\begin{align} \label{eq:mu_probability} \mu^\tr= \frac{1}{I^\tr} \frac{\prod_{i} a_{i}^\tr(\bb x)^{\Re \nu_i}}{\prod_{j} b_{j}^\tr(\bb{x})^{\Re \rho_j}} \Omega, \end{align}
such that $1 =\int_{\P_{>0}^{n-1}} \mu^\tr$. The integral in eq.~\eqref{eq:integral} can now be written as,
\begin{align*} I = I^\tr \int_{\P_{>0}^{n-1}} R_{a/b}(\bb x) \mu^\tr, \end{align*}
with $R_{a/b}$ as defined in eq.~\eqref{eq:def_residualab}. 
As $\mu^\tr$ is a properly normalized probability distribution on $\P^{n-1}_{>0}$, we 
can use Theorem~\ref{thm:montecarlo} to get a direct estimation algorithm for $I$ from this, provided that we have a reasonably efficient way to sample from the distribution $\mu^\tr$. %
\begin{algorithm}[H]
\begin{algorithmic}
\For{$\ell \in 1,\ldots, N$}
\State{Generate a random sample $\bb x^{(\ell)} \in \P^{n-1}_{>0}$ distributed as $\mu^\tr$ from eq.~\eqref{eq:mu_probability}.}
\EndFor
\State{Return $I^{(N)} = \frac{I^\tr}{N} \sum_{\ell=1}^N R_{a/b}(\bb x^{(\ell)})$.}
\end{algorithmic}
\caption{Monte Carlo quadrature using tropical sampling}
\label{alg:tropical_sampling}
\end{algorithm}

Algorithm~\ref{alg:tropical_sampling} is not obviously simpler or more efficient than Algorithm~\ref{alg:basic}, as the complicated part---generating a sample from the random distribution given by $\mu^\tr$---has been conveniently out-sourced. 

A simple method to sample from $\mu^\tr$ is to again use a geometric sector decomposition. By Theorem~\ref{thm:secdec} the tropically approximated integral in eq.~\eqref{eq:integral_trop} can be written as a sum,
\begin{align} \label{eq:def_Itr_secdec} I^{\tr} = \sum_{ \mathcal C \in \mathcal M_{\mathcal{AB}}^\Delta} I_{\mathcal C}^\tr \text{ with } I_{\mathcal C}^\tr &= \frac{ \left|\det(\bb u^{(\mathcal C,1)}, \ldots, \bb u^{(\mathcal C,n-1)},\one) \right| }{ \prod_{k=1}^{n-1} \inner{\bb u^{(\mathcal C,k)}}{\bb w^{(\mathcal C)}} }, \end{align}
where $I_{\mathcal C}^\tr > 0$ for all maximal cones $\mathcal C \in \mathcal M^\Delta_{\mathcal{AB}}$. Hence, we can interpret $I^\tr_\mathcal{C}/I^\tr$ as a probability assigned to each cone $\mathcal C \in \mathcal M^\Delta_{\mathcal{AB}}$ and draw a random cone accordingly. Drawing a random sample from a finite discrete probability distribution is a classic problem. It can be solved in constant time independent of the number of possible outcomes if a table of the probabilities of the respective outcomes is appropriately preprocessed, for instance by using the \emph{alias method} \cite[Section~3.4.1]{knuth2014art}. 
Provided that we have generated such a table together with a table of appropriate values of $\bb w^{(\mathcal C)}$ and $\bb u^{(\mathcal C,1)}, \ldots, \bb u^{(\mathcal C,n-1)}$ we can execute the following algorithm:
\begin{algorithm}[H]
\begin{algorithmic}
\State{Draw a random cone $\mathcal C \in \mathcal M^\Delta_{\mathcal{AB}}$ with probability $I^\tr_\mathcal{C}/I^\tr$.}
\State{Draw a random vector $\bb \xi \in [0,1]^{n-1}$ from the uniform distribution.}
\State{Set $x_k = \prod_{i = 1}^{n-1} \xi_i^{-u_k^{(\mathcal C,i)}/(\inner{\bb u^{(\mathcal C,i)}}{\bb w^{(\mathcal C)}})}$ for all $k\in \{1,\ldots,n\}$. }
\State{Return $\bb x=[x_1:\ldots:x_n] \in \Exp \mathcal C \subset \P^{n-1}_{>0}$ and $\mathcal C$.}
\end{algorithmic}
\caption{Algorithm to generate a sample with distribution $\mu^\tr$}
\label{alg:mu_sample}
\end{algorithm}

\begin{proposition}
Algorithm~\ref{alg:mu_sample} generates a sample $\bb x \in \P^{n-1}_{>0}$, distributed as $\mu^\tr$ in eq.~\eqref{eq:mu_probability}.
\end{proposition}
\begin{proof}
For any test function $f:\P_{>0}^{n-1} \rightarrow \C$  and a random sample $\bb x \in \P_{>0}^{n-1}$ generated by Algorithm~\ref{alg:mu_sample}, we have
\begin{align*} \E[ f(\bb x) ] = \sum_{\mathcal C \in \mathcal M^\Delta_{\mathcal{AB}}} I^\tr_\mathcal{C}/I^\tr \int_{[0,1]^{n-1}} f( \bb x^{(\mathcal C)}(\bb \xi)) \prod_{k=1}^{n-1} \dd \xi_k. \end{align*}
Using eq.~\eqref{eq:def_Itr_secdec} and Theorem~\ref{thm:secdec} gives 
\begin{gather*} \E[ f(\bb x) ] = \frac{1}{I^\tr} \int_{\P^{n-1}_{>0}} \frac{\prod_{i} a_{i}^\tr(\bb x)^{\Re \nu_i}}{\prod_{j} b_{j}^\tr(\bb{x})^{\Re \rho_j}} f(\bb x) \Omega = \int_{\P^{n-1}_{>0}} f(\bb x) \mu^\tr. \qedhere \end{gather*}
\end{proof}

To run both Algorithms~\ref{alg:tropical_sampling} and \ref{alg:mu_sample} together we need $N$ evaluations of the function $R_{a/b}(\bb x)$. Equivalently, we need proportional to $\delta^{-2}$ evaluations to obtain an estimate $I$ of $\delta$ accuracy. This is a significant improvement over Algorithm~\ref{alg:basic} as the runtime is now independent of the number of sectors $|\mathcal M^\Delta_{\mathcal{AB}}|$. %

It has to be stressed that this suggested direct comparison between Algorithm~\ref{alg:basic} and the combination of the Algorithms~\ref{alg:tropical_sampling} and \ref{alg:mu_sample} is flawed by the inherent difference in the respective proportionality factors for $\delta^{-2}$ or equivalently, in the number of samples $N$ that results in a given accuracy. In a situation, in which the sector integrals all contribute roughly the same value to the overall integral, Algorithms~\ref{alg:tropical_sampling} and \ref{alg:mu_sample} offer no advantage over Algorithm~\ref{alg:basic}. For practical applications the values of the sector integrals tend to differ heavily in magnitude, which makes Algorithms~\ref{alg:tropical_sampling} and \ref{alg:mu_sample} favorable.

Just as for Algorithm~\ref{alg:basic} a preprocessing step needs to be performed for Algorithms~\ref{alg:tropical_sampling} and \ref{alg:mu_sample}: the triangulation $\mathcal M^\Delta_{\mathcal{AB}}$ and the associated table needs to be calculated. %
This computation is also necessary to compute the normalization factor $I^\tr=\sum_{\mathcal C\in\mathcal M^\Delta_{\mathcal{AB}}} I_{\mathcal C}^\tr$. In the best case, the time it takes to create such a table will be proportional to the number of sectors $|\mathcal M^\Delta_{\mathcal{AB}}|$, but we only need to compute this table once and can evaluate an arbitrary large number of samples afterwards. 

Therefore, even though we are still effectively constrained by the dimension of the problem, which has to be small enough for the preprocessing step to be finished in a reasonable time, this constraint on the dimension is decoupled from the achievable accuracy. %

Recall that so far, we considered completely general integrals in eq.~\eqref{eq:integral}. Although we already managed to accelerate the integration for the general case in comparison to the traditional approach, further improvements are possible if more specific properties of the integrand are used. Especially, integrals that come from physical applications are well-known to carry a very rich geometric structure, whose exploitation offers a whole new set of tools to improve numerical approximation methods. In the following, we will achieve a further improvement in runtime, memory requirement and overall complexity by using a specific structure which is exhibited by a large family of integrals. Integrals of this family appear in many contexts in high energy physics. This family consists of all integrals as in eq.~\eqref{eq:integral} where the Newton polytopes of the polynomials $\{a_i\}$ and $\{b_j\}$ are \emph{generalized permutahedra}.

\section{Generalized permutahedra}
\label{sec:genperm}

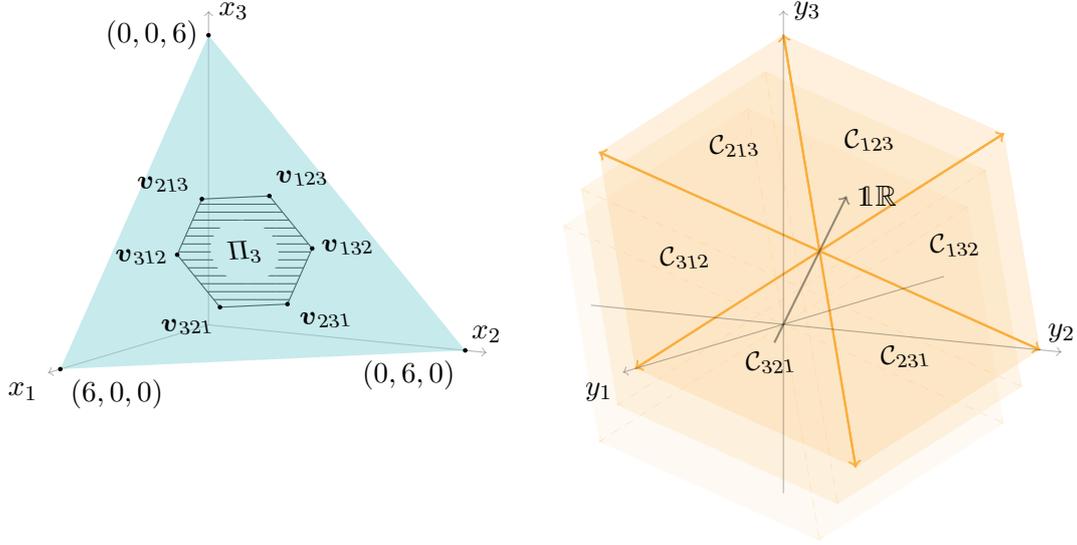
\begin{figure}

\tdplotsetmaincoords{80}{120}

\tikzset{perspective/.style= {canvas is plane={O(0,0,0)x(-1*0.70710678118,1*0.70710678118,0)y(-1*0.57735026919*0.70710678118,-1*0.57735026919*0.70710678118,2*0.57735026919*0.70710678118)}} }

\begin{subfigure}[t]{0.48\textwidth}
        \centering

\begin{tikzpicture}[tdplot_main_coords,scale=.65] \draw[thin, opacity = .3,->] (0,0,0) -- (6.5,0,0) node[anchor=north east,opacity=1]{$x_1$}; \draw[thin, opacity = .3,->] (0,0,0) -- (0,6.5,0) node[anchor=south,opacity=1]{$x_2$}; \draw[thin, opacity = .3,->] (0,0,0) -- (0,0,6.5) node[anchor=west,opacity=1]{$x_3$}; \coordinate (v123) at (1,2,3); \coordinate (v132) at (1,3,2); \coordinate (v213) at (2,1,3); \coordinate (v231) at (2,3,1); \coordinate (v312) at (3,1,2); \coordinate (v321) at (3,2,1); \draw[pattern=horizontal lines] (v123) -- (v132) -- (v231) -- (v321) -- (v312) -- (v213) -- (v123); \node[perspective](P) at (2,2,2) [circle,fill=white] {$\phantom{\Pi_3}$}; \draw[draw = none, fill = Aquamarine!50, opacity = 0.5] (6,0,0) -- (0,6,0) -- (0,0,6) -- (6,0,0); \filldraw (0,0,6) circle(1pt) node[left] {$(0,0,6)$}; \filldraw (0,6,0) circle(1pt) node[below left] {$(0,6,0)$}; \filldraw (6,0,0) circle(1pt) node[below right] {$(6,0,0)$}; \filldraw (v123) circle(1pt) node[above right,perspective] {$\bb v_{\text{123}}$}; \filldraw (v132) circle(1pt) node[right,perspective] {$\bb v_{\text{132}}$}; \filldraw (v213) circle(1pt) node[above left,perspective] {$\bb v_{\text{213}}$}; \filldraw (v231) circle(1pt) node[below right,perspective] {$\bb v_{\text{231}}$}; \filldraw (v312) circle(1pt) node[left,perspective] {$\bb v_{\text{312}}$}; \filldraw (v321) circle(1pt) node[below left,perspective] {$\bb v_{\text{321}}$}; \node[perspective] (P) at (2,2,2) {$\Pi_3$}; \foreach \i in {2,...,0} { \begin{scope}[shift={(-\i/3*3,-\i/3*3,-\i/3*3)}] \coordinate (O) at (2,2,2); \coordinate (v1) at (-2,4,4); \coordinate (v2) at (0,0,6); \coordinate (v3) at (4,-2,4); \coordinate (v4) at (6,0,0); \coordinate (v5) at (4,4,-2); \coordinate (v6) at (0,6,0); \draw[draw = none, fill = none] (v1) -- (v2) -- (v3) -- (v4) -- (v5) -- (v6) -- (v1); \end{scope} } \end{tikzpicture}
    \caption{The permutahedron $\Pi_3 \subset \R^3$ which is contained in the hyperplane $x_1+x_2+x_3 = 6$ as indicated.}
    \label{fig:permutahedron3}
    \end{subfigure}%
    \hfill
    \begin{subfigure}[t]{0.48\textwidth}
\begin{tikzpicture}[tdplot_main_coords,scale=.65] \begin{scope}[shift={(-2,-2,-2)}] \coordinate (O) at (2,2,2); \coordinate (v1) at (-2,4,4); \coordinate (v2) at (0,0,6); \coordinate (v3) at (4,-2,4); \coordinate (v4) at (6,0,0); \coordinate (v5) at (4,4,-2); \coordinate (v6) at (0,6,0); \draw[dashed,color=YellowOrange, opacity = .1,<-] (v2) -- (O) -- (v1); \draw[dashed,color=YellowOrange, opacity = .1,<-] (v3) -- (O) -- (v2); \draw[dashed,color=YellowOrange, opacity = .1,<-] (v4) -- (O) -- (v3); \draw[dashed,color=YellowOrange, opacity = .1,<-] (v5) -- (O) -- (v4); \draw[dashed,color=YellowOrange, opacity = .1,<-] (v6) -- (O) -- (v5); \draw[dashed,color=YellowOrange, opacity = .1,<-] (v1) -- (O) -- (v6); \draw[draw = none, fill = YellowOrange!50, opacity = .1] (v1) -- (v2) -- (v3) -- (v4) -- (v5) -- (v6) -- (v1); \end{scope} \begin{scope}[shift={(-1,-1,-1)}] \coordinate (O) at (2,2,2); \coordinate (v1) at (-2,4,4); \coordinate (v2) at (0,0,6); \coordinate (v3) at (4,-2,4); \coordinate (v4) at (6,0,0); \coordinate (v5) at (4,4,-2); \coordinate (v6) at (0,6,0); \draw[dashed,color=YellowOrange, opacity = .1,<-] (v2) -- (O) -- (v1); \draw[dashed,color=YellowOrange, opacity = .1,<-] (v3) -- (O) -- (v2); \draw[dashed,color=YellowOrange, opacity = .1,<-] (v4) -- (O) -- (v3); \draw[dashed,color=YellowOrange, opacity = .1,<-] (v5) -- (O) -- (v4); \draw[dashed,color=YellowOrange, opacity = .1,<-] (v6) -- (O) -- (v5); \draw[dashed,color=YellowOrange, opacity = .1,<-] (v1) -- (O) -- (v6); \draw[draw = none, fill = YellowOrange!50, opacity = .2] (v1) -- (v2) -- (v3) -- (v4) -- (v5) -- (v6) -- (v1); \end{scope} \begin{scope} \coordinate (O) at (2,2,2); \coordinate (v1) at (-2,4,4); \coordinate (v2) at (0,0,6); \coordinate (v3) at (4,-2,4); \coordinate (v4) at (6,0,0); \coordinate (v5) at (4,4,-2); \coordinate (v6) at (0,6,0); \draw[thick,color=YellowOrange, opacity = 1,<-] (v2) -- (O) -- (v1); \draw[thick,color=YellowOrange, opacity = 1,<-] (v3) -- (O) -- (v2); \draw[thick,color=YellowOrange, opacity = 1,<-] (v4) -- (O) -- (v3); \draw[thick,color=YellowOrange, opacity = 1,<-] (v5) -- (O) -- (v4); \draw[thick,color=YellowOrange, opacity = 1,<-] (v6) -- (O) -- (v5); \draw[thick,color=YellowOrange, opacity = 1,<-] (v1) -- (O) -- (v6); \draw[draw = none, fill = YellowOrange!50, opacity = .3] (v1) -- (v2) -- (v3) -- (v4) -- (v5) -- (v6) -- (v1); \end{scope} \draw[thick,color=black, opacity = .4,->] (-.5,-.5,-.5) -- (3.5,3.5,3.5) node[anchor=west,opacity=1] {$\one \R$}; \draw[thin, opacity = .3,->] (-6.5,0,0) -- (6.5,0,0) node[anchor=north east,opacity=1]{$y_1$}; \draw[thin, opacity = .3,->] (0,-4.5,0) -- (0,6.5,0) node[anchor=south,opacity=1]{$y_2$}; \draw[thin, opacity = .3,->] (0,0,-3.5) -- (0,0,6.5) node[anchor=west,opacity=1]{$y_3$}; \begin{scope} \node[perspective](C123) at (0,2,4) [circle] {$\mathcal C_{\text{123}}$}; \node[perspective](C213) at (2,0,4) [circle] {$\mathcal C_{\text{213}}$}; \node[perspective](C231) at (2,4,0) [circle] {$\mathcal C_{\text{231}}$}; \node[perspective](C312) at (4,0,2) [circle] {$\mathcal C_{\text{312}}$}; \node[perspective](C321) at (4,2,0) [circle] {$\mathcal C_{\text{321}}$}; \node[perspective](C132) at (0,4,2) [circle] {$\mathcal C_{\text{132}}$}; \end{scope} \end{tikzpicture}
    \caption{The braid arrangement fan $\mathcal F_{\Pi_3}/\one \R$ which partitions $\R^3/\one \R$ with equivalent hyperplanes orthogonal to $\one \R$ indicated.}
    \label{fig:braid_arrangement_fan}
    \end{subfigure}%
    \caption{The permutahedron $\Pi_3$ and its reduced normal fan. Vertices and maximal cones are both labelled by the associated permutations.}
    \label{fig:permutahedron_fan}
\end{figure}

The permutahedron $\Pi_n$ is an ($n-1$)-dimensional polytope in $\R^n$. It can be defined as the convex hull of $n!$ vertices determined by permutations in $S_n$:
\begin{align*} \Pi_n = \left\{ \sum_{\sigma \in S_n} \lambda_\sigma \bb v^{(\sigma)} : \sum_{\sigma \in S_n} \lambda_{\sigma} = 1 \text{ and } \lambda_{\sigma} \geq 0 \right\} \subset \R^n, \end{align*}
where the vector $\bb v^{(\sigma)} = (\sigma(1), \ldots, \sigma(n)) \in \R^n$ encodes the permutation $\sigma$. The permutahedron is contained in the hyperplane $\Pi_n \subset \{ \bb v \in \R^n: \inner{\one}{\bb v} = n(n+1)/2 \}$ and is full-dimensional within this hyperplane. The permutahedron $\Pi_3$ is depicted in Figure~\ref{fig:permutahedron3}.
The cones of maximal dimension in the reduced normal fan $\mathcal F_{\Pi_n}/\one \R$ of $\Pi_n$ are labelled by permutations as well. They are of the form,
\begin{align} \label{eq:CsigmaWeyl} \mathcal C_\sigma = \{ \bb y \in \R^n/\one\R: y_{\sigma(1)} \leq \ldots \leq y_{\sigma(n)} \} \end{align}
such a domain is called a \emph{Weyl chamber}. 
It is not hard to see that these are simplicial cones as
\begin{align} \label{eq:CsigmaWeylexplicit} \mathcal C_\sigma = \left\{ \sum_{k = 1}^{n-1} \lambda_k \bb u^{(\sigma,k)}: \lambda_k \geq 0 \right\} \text{ with } u^{(\sigma,k)}_{\sigma(i)} = \begin{cases} -1 &\text{if $k \leq i$} \\ 0 &\text{else} \end{cases} \end{align}
where we chose the set of representatives in $\R^n$ of the vectors in $\bb u^{(\sigma,k)} \in \R^n/\one \R$ by fixing $u^{(\sigma,k)}_{\sigma(n)} = 0$ for all $k\in\{1,\ldots,n-1\}$. The remaining cones of the reduced normal fan $\mathcal F_{\Pi_n}$ can be constructed by taking arbitrary intersections of these cones. This fan is also called the \emph{braid arrangement fan}. The reduced normal fan of $\Pi_3$ is illustrated in Figure~\ref{fig:braid_arrangement_fan}.

\begin{definition}[Generalized permutahedron {\cite[Definition~6.1]{postnikov2009permutohedra}}]
\label{def:gen_permutahedron}
A polytope whose normal fan is a coarsening of $\mathcal F_{\Pi_n}$ is a generalized permutahedron.
\end{definition}
Generalized permutahedra have a large number of remarkable properties \cite{postnikov2009permutohedra,aguiar2017hopf}. E.g. %
\begin{theorem}[{\cite[Definition~6.1]{postnikov2009permutohedra} and \cite[Theorem~12.3]{aguiar2017hopf}}]
\label{thm:inequality_gp}
A generalized permutahedron $\mathcal G_z$ has the facet presentation 
\begin{align} \label{eq:inequality_gp} \mathcal G_z = \left\{ \bb v \in \R^n : \sum_{i \in [n]} v_i = z({[n]}) \text{ and } \sum_{i \in I} v_i \geq z({I}) \text{ for all } I \subset [n] \right\}, \end{align}
where $[n] = \{1,\ldots, n\}$ and $z$ is a \emph{supermodular boolean function} $z: \bb 2^{[n]} \rightarrow \R$ with $z(\emptyset) = 0$. In fact every supermodular boolean function, that means $z: \bb 2^{[n]} \rightarrow \R$ with 
\begin{align*} z(A) + z(B) \leq z(A \cap B) + z(A \cup B ) \text{ for all } A,B \subset [n], \end{align*}
gives rise to a generalized permutahedron by the inequality description in eq.~\eqref{eq:inequality_gp}\footnote{We are using a different sign notation than \cite{aguiar2017hopf}, but agree with \cite{postnikov2009permutohedra}.}.
\end{theorem}

\begin{corollary}
\label{crll:gp_inclusion}
If both $\mathcal G_{z_1}$ and $\mathcal G_{z_2}$ are generalized permutahedra and 
$z_1(A) > z_2(A)$ for all non-empty $A \subsetneq [n]$ and $z_1([n]) = z_2([n])$, then
$\mathcal G_{z_1} \subset \relint \mathcal G_{z_2}$.
\end{corollary}
\begin{proof}
By Theorem~\ref{thm:inequality_gp} it follows immediately that 
$\mathcal G_{z_1} \subset \mathcal G_{z_2}$.
The inequalities in eq.~\eqref{eq:inequality_gp} are strict \cite[Theorem~12.3]{aguiar2017hopf}. Therefore the statement follows.
\end{proof}

The Minkowski sum of two generalized permutahedra is again a generalized permutahedron (see for instance \cite[Lemma~2.2.2]{doker2011geometry}):
\begin{lemma}
\label{lmm:gpMinkowski}
If both $\mathcal G_{z_1}$ and $\mathcal G_{z_2}$ are generalized permutahedra, then also their Minkowski sum 
$ \mathcal G_{z_{12}} = \mathcal G_{z_1} + \mathcal G_{z_2} $ is a generalized permutahedron with the boolean functions $z_1,z_2,z_{12}:\bb 2^{[n]} \rightarrow \R$ related by $z_{12}(A) = z_1(A) + z_2(A)$ for all $A \subset [n] = \{1,\ldots, n\}$.
\end{lemma}

A vector $\bb v\in \mathcal G_z$ which maximizes all linear functionals in a Weyl chamber $\mathcal C_\sigma$ is a vertex of $\mathcal G_z$. 
This gives a canonical map from permutations $\sigma \in S_n$ to the vertices of a generalized permutahedron. We can use a result of Fujishige and Tomizawa to explicitly construct this map:

\begin{lemma}[{\cite[Lemma~3.1, Lemma~3.2]{fujishige1983note}}]
\label{lmm:winG}
If $z$ is a supermodular function $z: \bb 2^{[n]} \rightarrow \R$, $\sigma\in S_n$ a permutation and $\bb w^{(\sigma,z)}\in \R^n$ is the vector given component-wise by
\begin{align} \label{eq:wsigma_vector} w^{(\sigma,z)}_{\sigma(k)} = z(A_k^\sigma) - z(A_{k-1}^\sigma) \text{ for all } k \in [n], \end{align}
where $A_k^\sigma = \{\sigma(1), \ldots, \sigma(k)\} \subset [n]=\{1,\ldots,n\}$, then
$\bb w^{(\sigma,z)}$ is a vertex of the generalized permutahedron $\mathcal G_z$ and
\begin{gather*} \inner{\bb y}{\bb w^{(\sigma,z)}} = \max_{\bb v \in \mathcal G_z} \inner{\bb y}{\bb v}\text{ for all } \bb y \in \mathcal C_\sigma, \end{gather*}
where $\mathcal C_\sigma$ is a Weyl-chamber in the braid arrangement fan as defined in eq.~\eqref{eq:CsigmaWeyl}. 
\end{lemma}
\subsection{Tropical sampling for generalized permutahedra}
In general, it is necessary to compute a triangulation of the reduced refined normal fans of the $\mathcal A$ and $\mathcal B$ polytopes to perform the procedure described in Section~\ref{sec:trop_sampling}. This cumbersome computation can be circumvented if the $\mathcal A$ and $\mathcal B$ polytopes are generalized permutahedra. In this case, there is an especially simple way to sample from the associated tropical measure $\mu^\tr$ defined in eq.~\eqref{eq:mu_probability} without the need for an explicit triangulation as required for Algorithm~\ref{alg:mu_sample}.

From now on, we will therefore assume that the polytopes $\mathcal A$ and $\mathcal B$ are both generalized permutahedra. This implies by Theorem~\ref{thm:inequality_gp} that there are unique boolean functions $z_{\mathcal A} : \bb 2^{[n]} \rightarrow \R$ and $z_{\mathcal B}$ analogously which describe these polytopes. Using these functions and the properties of generalized permutahedra introduced above, we can state

\begin{theorem}[Geometric sector decomposition for generalized permutahedra]
\label{thm:secdec_gp}
Let $\mathcal A$ and $\mathcal B$ be the polytopes defined in Theorem~\ref{thm:convergence}. %
If $\mathcal A$ and $\mathcal B$ are generalized permutahedra with associated boolean functions $z_{\mathcal A},z_{\mathcal B}: \bb 2^{[n]} \rightarrow \R$ which 
fulfill the requirements \normalfont{R1} and \normalfont{R2} of Theorem~\ref{thm:convergence},
then we can write the integral 
\begin{gather*}  I[f] = \int_{\P_{>0}^{n-1}} \frac{\prod_{i} a_{i}^\tr(\bb x)^{\Re \nu_i}}{\prod_{j} b_{j}^\tr(\bb{x})^{\Re \rho_j}} f(\bb x) \Omega      \end{gather*}
as a sum 
$I[f] = \sum_{\sigma \in S_n} I_{\sigma}[f]$ with
\begin{align*}  I_{\sigma}[f] &= \frac{1} { \prod_{k=1}^{n-1} r(A^\sigma_k) } \int_{[0,1]^{n-1}} f \left( \bb x^{(\sigma)}(\bb \xi) \right) \prod_{k=1}^{n-1} \dd \xi_k, \end{align*}
where 
\begin{itemize}
\item $f:\P_{>0}^{n-1} \rightarrow \C$ is a measurable homogeneous function of degree $0$,
\item 
$A_k^\sigma = \{\sigma(1), \ldots, \sigma(k)\} \subset [n]=\{1,\ldots,n\}$,
\item $r(A) = z_{\mathcal A}(A) - z_{\mathcal B}(A) $,
which fulfills $r(A) > 0$ for all non-empty proper subsets $A \subsetneq [n]$ and
\item
$\bb x^{(\sigma)}(\bb \xi) \in \Exp(\mathcal C_\sigma)$ is given component-wise by $x_{\sigma(k)} = \prod_{i=k}^{n-1} \xi_i^{1/r(A^\sigma_i)}$ and $x_{\sigma(n)} = 1$.
\end{itemize}
\end{theorem}
\begin{proof}
This theorem is a specialization of Theorem~\ref{thm:secdec} to the generalized permutahedron case. The braid arrangement fan defined in eq.~\eqref{eq:CsigmaWeyl} provides an appropriate reduced simplicial fan.
By Lemma~\ref{lmm:winG} we have vertices $\bb w^{(\sigma,z_{\mathcal A})} \in \mathcal A$ and $\bb w^{(\sigma,z_{\mathcal B})} \in \mathcal B$ such that
$ \inner{\bb y}{\bb w^{(\sigma,z_{\mathcal A})}} = \max_{\bb v \in \mathcal A} \inner{\bb y}{\bb v} $
and 
$ \inner{\bb y}{\bb w^{(\sigma,z_{\mathcal B})}} = \max_{\bb v \in \mathcal B} \inner{\bb y}{\bb v} $
for all $\sigma \in S_n$ and $\bb y \in \mathcal C_\sigma$.
Using the explicit representatives of the generators $\bb u^{(\sigma,1)}, \ldots,\bb u^{(\sigma,n-1)}$ of the cone $\mathcal C_\sigma$ from eq.~\eqref{eq:CsigmaWeylexplicit} together with Lemma~\ref{lmm:winG} gives
$\inner{\bb u^{(\sigma,k)}}{\bb w^{(\sigma,z_{\mathcal A})}} = -z_{\mathcal A}(A_k^\sigma)$ and
$\inner{\bb u^{(\sigma,k)}}{\bb w^{(\sigma,z_{\mathcal B})}} = -z_{\mathcal B}(A_k^\sigma)$.  It follows from this and Lemma~\ref{lmm:ABsupportlimit} that $\inner{\bb u^{(\sigma,k)}}{\bb w^{(\sigma,z_{\mathcal B})}} - \inner{\bb u^{(\sigma,k)}}{\bb w^{(\sigma,z_{\mathcal A})}} = z_{\mathcal A}(A_k^\sigma) - z_{\mathcal B}(A_k^\sigma) > 0$ for all $\sigma \in S_n$ and $k\in \{1,\ldots,n-1\}$ which implies $r(A) > 0$ for all non-empty $A\subsetneq [n]$.
From the form of the $\bb u^{(\sigma,k)}$ vectors in eq.~\eqref{eq:CsigmaWeylexplicit} it is obvious that 
$|\det(\bb u^{(\sigma,1)}, \ldots,\bb u^{(\sigma,n-1)},\one)|=1$.
\end{proof}

Theorem~\ref{thm:secdec_gp} ensures that we can proceed as above and perform the Monte Carlo Algorithms~\ref{alg:tropical_sampling} and \ref{alg:mu_sample} just as in the general case. It is clear that the preprocessing step will be straightforward as generalized permutahedra come with an appropriately simplicial fan `built in'. Algorithm~\ref{alg:mu_sample} requires us to generate a table of size $n!$ as we need one entry for each cone in the braid arrangement fan.  For this algorithm to be applicable in a computationally feasible way that table needs to be stored in the memory of the computer. Hence the naive algorithm is only practically applicable for relatively small values of $n$. 

However, a further significant improvement can be achieved: it is not necessary to store an entry for each permutation in a table. If a small additional computation for each sampled point is performed, a table of size proportional to $2^{n}$ suffices. We will describe this specialized version of Algorithm~\ref{alg:mu_sample} in the rest of this section.

First observe that the overall normalization factor needed to apply Algorithm~\ref{alg:mu_sample} is given by 
\begin{align} \label{eq:ItrdefGz} I^\tr &= \sum_{\sigma \in S_n} I^\tr_{\mathcal C_\sigma} \text{ with } I^\tr_{\mathcal C_\sigma} = \frac{1} { \prod_{k=1}^{n-1} r(A^\sigma_k) }, \end{align}
where $r(A) = z_{\mathcal A}(A) - z_{\mathcal B}(A) $ for all non-empty $A \subsetneq [n]$. This equation is just eq.~\eqref{eq:def_Itr_secdec} specified using Theorem~\ref{thm:secdec_gp} to the generalized permutahedron case. 
For the following considerations it will be convenient to declare $r(\emptyset) = 1$, which opens the way towards the following generalization that promotes $I^\tr$ to a boolean function on $\bb 2^{[n]}$:
\begin{definition}
\label{def:recursionJ}
For a boolean function $r:\bb 2^{[n]} \rightarrow \R$ with $r(\emptyset) = 1$ and $r(A) > 0$ for all non-empty $A \subsetneq [n]$, we define the boolean function $J_{r}: \bb 2^{[n]} \rightarrow \R_{>0}$ recursively as
\begin{align*}  J_{r}(A) = \sum_{e\in A} \frac{J_r(A\setminus e)}{ r(A\setminus e) } \text{ for all non-empty } A \subset [n] \text{ where } J_r(\emptyset) = 1. \end{align*}
\end{definition}
\begin{proposition}
\label{prop:Jrspecial}
If $r(A) = z_{\mathcal A}(A) - z_{\mathcal B}(A) $ for all non-empty $A \subsetneq [n]$ and $r(\emptyset) =1$, then
$I^\tr = J_r([n])$.
\end{proposition}

\begin{proof}
We will prove that 
$ J_r(A) = \sum_{\sigma : [m] \rightarrow A} \frac{1} { \prod_{k=1}^{m-1} r(A^\sigma_k) } $, where the sum is over all bijections $\sigma: [m] \rightarrow A$.
Fixing such a bijection is equivalent to fixing a pair $(e,\mu)$ of an element $e \in A$ and a bijection $\mu:[m-1] \rightarrow A\setminus e$. Decomposing the sum in this way and using eq.~\eqref{eq:ItrdefGz} gives the statement.
\end{proof}
\begin{remark}
This recursive method to calculate the normalization factor $I^\tr$ might also be useful in other contexts. For instance, this can be used to calculate the volume of the \emph{polar dual} of a generalized permutahedron fairly efficiently. In the context of scattering amplitudes this recursion can also be used to calculate the canonical form of a generalized permutahedron.
\end{remark}

If we prepare a table of the values $J_r(A)$ and $r(A)$ for all $A \subset [n]$, we can run the following algorithm:
\begin{algorithm}[H]
\begin{algorithmic}[0]
\State{Set $A = [n]$ and $\kappa = 1$.}
\While{$A \neq \emptyset$}
\State{Pick a random $e \in A$ with probability $p_e = \frac{1}{J_r(A)} \frac{J_r(A\setminus e)}{r(A\setminus e)}$.}
\State{Remove $e$ from $A$, i.e.~set $A \gets A \setminus e$.}
\State{Set $\sigma({|A|}) = e$.}
\State{Set $x_e = \kappa$.}
\State{Pick a uniformly distributed random number $\xi \in [0,1]$.}
\State{Set $\kappa \gets \kappa \xi^{1/r(A)}$.}
\EndWhile
\State{Return $\bb x = [x_1,\ldots,x_n] \in \Exp(\mathcal C_\sigma) \subset \P^{n-1}_{>0}$ and $\sigma = (\sigma(1),\ldots,\sigma(n)) \in S_n$.}
\end{algorithmic}
\caption{to generate a sample from $\mu^\tr$ for generalized permutahedra}
\label{alg:gp_sampling}
\end{algorithm}
Note that the probability distribution $p_e = \frac{1}{J_r(A)} \frac{J_r(A\setminus e)}{r(A\setminus e)}$ over the elements $e\in A$ is properly normalized due to Definition~\ref{def:recursionJ}.

\begin{proposition}
If $r(A) = z_{\mathcal A}(A) - z_{\mathcal B}(A) $ for all non-empty $A \subsetneq [n]$, $r(\emptyset) =1$ and $J_r$ is the boolean function given in Definition~\ref{def:recursionJ}, then
Algorithm~\ref{alg:gp_sampling} generates a sample $\bb x \in \P^{n-1}_{>0}$, distributed as $\mu^\tr$ in eq.~\eqref{eq:mu_probability} in the generalized permutahedron case.
\end{proposition}
\begin{proof}
For any test function $f:\P_{>0}^{n-1} \rightarrow \C$ and a random sample $\bb x \in \P_{>0}^{n-1}$ generated by Algorithm~\ref{alg:gp_sampling}, 
\begin{gather*} \E[ f(\bb x) ] = \sum_{e_n \in A_n} \frac{1}{J_r(A_n)} \frac{J_r(A_n\setminus e_n)}{r(A_n\setminus e_n)}   \ldots \sum_{e_{1} \in A_{1}} \frac{1}{J_r(A_{1})} \frac{J_r(A_{1}\setminus e_{1})}{r(A_{1}\setminus e_{1})} \int_{[0,1]^{n-1}}f(\bb x(\xi)) \prod_{k=1}^{n-1}\dd \xi_k, \end{gather*}
where we gave distinguished subscripts to the numbers $e$ and sets $A$ in the reverse order in which they appear in Algorithm~\ref{alg:gp_sampling} and $\bb x(\xi)$ is component-wise $x_{e_k} = \prod_{i=k}^{n-1} \xi_i^{1/r(A_i)}$. We identify $A_k \setminus e_k = A_{k-1}$. The terms $J_r(A_{k}\setminus e_{k})$ telescope, $J_r(\emptyset) = r(\emptyset) = 1$ and we get
\begin{align*} \E[ f(\bb x) ] = \frac{1}{J_r(A_n)} \sum_{e_n \in A_n} \ldots \sum_{e_{1} \in A_{1}} \frac{1}{r(A_n) \cdots r(A_1)}   \int_{[0,1]^{n-1}}f(\bb x(\xi)) \prod_{k=1}^{n-1}\dd \xi_k. \end{align*}
The sum can be written as a sum over all permutations in $\sigma \in S_n$ and $A_k = A_k^\sigma$.
The statement follows from Proposition~\ref{prop:Jrspecial} and Theorem~\ref{thm:secdec_gp}.
\end{proof}

Algorithm~\ref{alg:gp_sampling} allows us to integrate any integral of the form in eq.~\eqref{eq:integral} via Monte Carlo quadrature without actually performing any complicated triangulation or non-trivial sector decomposition step if the polytopes $\mathcal A$ and $\mathcal B$ are generalized permutahedra. Compared to the naiver approach where a table of size $n!$ is needed, only a table of size $2^n$ is required. The complexity of the preprocessing step is similarly reduced as the recursion in Definition~\ref{def:recursionJ} gives an efficient way to calculate all the necessary constants: The table for $J_r(A)$ can be calculated in $\bigO(n 2^n)$ steps. 

All this achieves not only a huge improvement in the required runtime and memory of the algorithm, but also significantly reduces the complexity of the overall algorithm. Triangulating an $n$-dimensional polytope is an involved algorithmic task. Circumventing this triangulation with the approach above makes it straightforward to implement an efficient integration algorithm. A detailed example is given in the following section. 

\section{Feynman integrals}
\label{sec:feynman}

A scalar Feynman integral associated to a Feynman graph $G$ with $E$ edges and $V$ vertices in parametric representation in $D$-dimensional Euclidean space can be written as,
\begin{align} \label{eq:parametric} I_G &= \int_{\P^{E-1}_{> 0}} \frac{\prod_{e} x_e^{\nu_e}}{\Psi_G(\bb x)^{D/2}} \left( \frac{\Psi_G(\bb x)}{\Phi_G(\bb x)} \right)^{\omega(G)} \Omega, \end{align}
which depends on the \emph{edge weights} $\nu_1, \ldots, \nu_E$, which we will assume to be positive and real. The \emph{superficial degree of divergence} $\omega(G)$ is given by $\omega(G) = \sum_{e} \nu_e - \ell(G) D/2$, where $\ell(G)$ is the number of loops of $G$ (i.e.~the first Betti number of $G$). The \emph{Kirchhoff-Symanzik polynomials} $\Psi_G$ and $\Phi_G$ are homogeneous of degree $\ell(G)$ and $\ell(G)+1$ in the $x_e$ variables. Obviously, the integral $I_G$ is a specific instance of an integral of the form in eq.~\eqref{eq:integral}. To simplify the notation, we omitted a prefactor of $\Gamma(\omega_G)/\prod_e \Gamma(\nu_e)$, which is usually included in the definition of scalar Feynman integrals. See for instance \cite{Nakanishi:110324} for details on this representation of Feynman integrals.

A complete finiteness proof of the Euclidean space Feynman integral $I_G$ together with an analysis of its analytic continuation properties in the $\nu_1,\ldots,\nu_E$ parameters has been achieved by Speer~\cite{Speer:1975dc}. More recently, Brown~\cite{Brown:2015fyf} showed that there is a canonical way to associate the integral $I_G$ to a \emph{motivic} avatar, which can be thought of as a specific representation of a conjectured \emph{cosmic Galois group}. This group suggests the existence of a \emph{coaction principle} which relates different Feynman integrals in a highly non-trivial way and it allows to analyse Feynman integrals with a whole new toolkit of homological methods and representation theory. 

For our endeavour to merely evaluate the integrals $I_G$, we can make use of parts of Brown's analysis \cite{Brown:2015fyf} to ensure that the relevant polytopes associated to the $\Psi_G$ and $\Phi_G$ polynomials are generalized permutahedra. We will start by giving some additional details on these polynomials including an efficient way to evaluate them.

\subsection{Symanzik polynomials}
Explicitly, the polynomials can be expressed as sums over \emph{spanning trees} $T_1$ and \emph{spanning $2$-forests} (spanning forests with two connected components) $T_2$ of the graph $G$:
\begin{align} \label{eq:PsiPhi_slow} \Psi_G(\bb x) &= \sum_{T_1} \prod_{e\notin T_1} x_e & \Phi_G(\bb x) &= \sum_{T_2} \bb \|\bb p(T_2)\|^2 \prod_{e\notin T_2} x_e + \Psi_G \sum_{e} x_e m_e^2, \end{align}
where  $\bb p(T_2)$ is the total momentum flowing between the two components of the $2$-forest $T_2$. Only the $\Phi_G(\bb x)$ polynomial depends on the external physical parameters: a set of \emph{momenta} $\bb p^{(1)}, \ldots, \bb p^{(V)} \in \R^D$ incoming into each of the vertices and a set of masses $m_1, \ldots, m_E \in \R$ associated to the edges of the graph. These polynomials can also be written in terms of the weighted $V\times V$ Laplace matrix of the graph (see for instance \cite{Bogner:2010kv}), which is component-wise
\begin{align*} L_{v,w} &= \begin{cases} -x_e^{-1} &\text{if there is an edge $e$ between $v$ and $w$} \\ \sum_{e \text{ incident to } v} x_e^{-1} &\text{if $v=w$ } \\ 0 &\text{else} \end{cases} \end{align*}
This matrix is only positive semi-definite whereas the reduced Laplacian $\widetilde L_{v,w}(\bb x)$, which is given by an arbitrary leading principle minor of the matrix $L_{v,w}(\bb x)$, is positive definite. The Symanzik polynomials can be written as
\begin{align} \label{eq:PsiPhi_fast} \Psi_G(\bb x) &= \left(\prod_e x_e \right) \det( \widetilde L ) & \Phi_G(\bb x) &= \Psi_G\left( \Tr( P^T \widetilde L^{-1} P) + \sum_e x_e m_e^2 \right), \end{align}
where $P$ is the $(V-1) \times D$ matrix, given row-wise by the incoming momenta, $\bb p^{(v)} \in \R^D$: $P_{v,\mu} = p^{(v)}_\mu$ where $v=1,\ldots,V-1$ and $\mu$ is a $D$-dimensional spacetime index. Note that due to momentum conservation no information is lost when only $V-1$ of the $V$ incoming momenta are used.

The second representation of the Symanzik polynomials in eq.~\eqref{eq:PsiPhi_fast} is more suitable for numerical evaluation than eq.~\eqref{eq:PsiPhi_slow}. The number of spanning trees of a graph grows exponentially with the number vertices $V$ \cite{mckay1983spanning} and the evaluation of the expressions in eq.~\eqref{eq:PsiPhi_slow} quickly becomes intractable when the graph gets large. The evaluation of the determinant with the other matrix operations in eq.~\eqref{eq:PsiPhi_fast} is computationally much more favorable: With a Cholesky decomposition of the matrix $\widetilde L(\bb x)$ both the value of $\Psi_G(\bb x)$ and $\Phi_G(\bb x)$ can be immediately calculated. Computing the Cholesky decomposition of a $(V-1)\times (V-1)$-matrix takes $\bigO(V^3)$ time. Due to the special structure of the problem---the matrix $\widetilde L(\bb x)$ being the reduced Laplace matrix of a graph---there even exists a nearly linear time approximation algorithm \cite{spielman2014nearly} to compute this decomposition.

To give a precise account on the Newton polytopes of the Symanzik polynomials we need some additional notation from \cite{Brown:2015fyf} for subgraphs of Feynman graphs. A subgraph $\gamma \subset \Gamma$ is equivalent to a subset of edges of the graph $\Gamma$. The set of subgraphs is therefore isomorphic to the set $\bb 2^{[E]}$ and we will identify boolean functions $\bb 2^{[E]} \rightarrow \R$ with functions defined on the set of subgraphs of the graph $G$. Just as for $G$, we will denote the first Betti number of a subgraph (i.e.~the number of loops) as $\ell(\gamma)$. A subgraph $\gamma \subset G$ is called mass-momentum-spanning (m.m.)\ in $G$ if the second Symanzik polynomial of the contracted graph $G/\gamma$ vanishes $\Phi_{G/\gamma} = 0$. Mass-momentum-spanning graphs can also be defined combinatorially as subgraphs that contain all massive edges and one connected component which connects all vertices with non-zero incoming momentum. See \cite[Definition~2.6]{Brown:2015fyf} for details on these types of subgraphs.

\begin{theorem}
\label{thm:psi_phi_facet}
If we restrict to Euclidean and non-exceptional kinematics, then the Newton polytope of $\Psi_G$ and $\Phi_G$ are generalized permutahedra. A facet presentation of these polytopes is given by the supermodular functions 
\begin{align*} z_{\Psi_G}(\gamma) &= \ell(\gamma) \\ z_{\Phi_G}(\gamma) &= \begin{cases} \ell(\gamma) +1 &\text{ if $\gamma$ is m.m.\ in } G \\ \ell(\gamma) &\text{ else } \end{cases} \end{align*}
for all subgraphs $\gamma \subset \Gamma$.
\end{theorem}
See \cite[Section~1.7]{Brown:2015fyf} for a definition of \emph{non-exceptional} or \emph{generic} kinematics. Briefly, this condition ensures that there is 
no non-trivial combination of the external momenta that adds up to $0$. 
It is worth remarking that this condition is not necessary if the combinatorial concept of mass-momentum-spanning is slightly generalized while keeping the equivalence $\Phi_{G/\gamma} = 0 \Leftrightarrow \text{mass-momentum-spanning}$. With this generalization it is sufficient to 
require that \emph{any} external momenta are non-zero. %
\begin{proof}
Theorem~\ref{thm:psi_phi_facet} has been proven by Schultka \cite[Theorem~4.15]{schultka2018toric} using results from Brown \cite{Brown:2015fyf}. 

The first statement that the Newton polytopes of $\Psi_G$ and $\Phi_G$ are generalized permutahedra can be traced back to Hepp \cite{Hepp:1966eg} and Speer \cite{Speer:1975dc}, who realized that a complete ordering of the integration parameters in eq.~\eqref{eq:parametric} is sufficient to capture the relevant singularities of parametric integrals in the Euclidean non-exceptional case. See also \cite{Smirnov:2008aw} for a comparison of this viewpoint with modern sector decomposition techniques.

The form of the boolean functions $z_{\Psi_G}$ and $z_{\Phi_G}$ follows directly from the factorization laws \cite[Proposition~2.2]{Brown:2015fyf}, \cite[Proposition~2.4]{Brown:2015fyf} and \cite[Theorem~2.7]{Brown:2015fyf} of the $\Psi_G$ and $\Phi_G$ polynomials. Their supermodularity follows from the argument in \cite{schultka2018toric} after Corollary~4.12.
\end{proof}
\begin{remark}
It was implicitly proved by Panzer \cite[Lemma~2.8]{Panzer:2019yxl} that the Newton polytope of $\Psi_G$ is a generalized permutahedron using an elegant argument based on \emph{Kruskal's algorithm} \cite{kruskal1956shortest}. This argument can also be generalized to the $\Phi_G$ polynomials by an extension of Kruskal's algorithm to minimal $2$-forests. 
\end{remark}

\begin{remark}
Theorem~\ref{thm:psi_phi_facet} is also of interest in a different context: Generalized permutahedra have a universal property with respect to their \emph{Hopf monoid} structure. Feynman graphs carry a Hopf algebra structure which is deeply intertwined with renormalization \cite{Connes:1999yr} and encodes the singularity structure of the integrand \cite{Bloch:2005bh,Brown:2015fyf}. The relationship between these two structures remains to be explored.
\end{remark}

\begin{remark}
For general non-Euclidean kinematics, the Newton polytope of $\Phi_G$ is not a generalized permutahedron. An explicit counterexample is given in \cite[Section~2.4]{Smirnov:2012gma}. We emphasize that the general tropical sampling algorithm introduced in Section~\ref{sec:trop_sampling} still applies. The caveat is that an explicit triangulation has to be computed in contrast to the generalized permutahedron case where no explicit triangulation is necessary.  
\end{remark}

\subsection{Tropical Monte Carlo quadrature of Euclidean Feynman integrals}
To perform the generalized permutahedron tropical Monte Carlo routine from Section~\ref{sec:genperm} on the parametric Feynman integral in eq.~\eqref{eq:parametric}, we still have to ensure that the numerator monomials $\prod_e x_e^{\nu_e}$ are generalized permutahedra. This is of course trivial, as the Newton polytope of a monomial is zero-dimensional and its normal fan is trivial. The braid arrangement fan is automatically a refinement of this fan and the conditions for Definition~\ref{def:gen_permutahedron} are fulfilled. The facet presentation of these $0$-dimensional polytopes associated to the Newton polytope of the polynomial $p_e(\bb x) = x_e$ in the form of Theorem~\ref{eq:inequality_gp} is given by the boolean function $z_{p_e}(\gamma) = 1$ if $e \in \gamma$ and $z_{p_e}(\gamma) = 0$ if $e\not \in \gamma$ for all subgraphs $\gamma$.

Because we assume that the edge weights $\nu_e$, the dimension $D$ and the superficial degree of divergence $\omega(G)$ are real, 
we have
\begin{align*} \mathcal A &= \sum_{e} \nu_e \NP_{p_e} +~ \omega(G) \NP_{\Psi_G} & \mathcal B &= \frac12 D \NP_{\Psi_G} +~ \omega(G) \NP_{\Phi_G} &&\text{if $\omega(G) \geq 0$ and } \\ \mathcal A &= \sum_{e} \nu_e \NP_{p_e} +~ (-\omega(G)) \NP_{\Phi_G} & \mathcal B &= \frac12 D \NP_{\Psi_G} +~ (-\omega(G)) \NP_{\Psi_G} &&\text{if $\omega(G) < 0$.} \end{align*}
we can define the boolean function $r_G:\bb 2^{[E]} \rightarrow \R$ as in Theorem~\ref{thm:secdec_gp},
\begin{align*} r_G(\gamma) &= z_{\mathcal A}( \gamma ) - z_{\mathcal B}( \gamma ) \\ &= \sum_{e} \nu_e z_{p_e}(\gamma) - \frac{D}{2} z_{\Psi_G}(\gamma) + \omega(G) (z_{\Psi_G}(\gamma) - z_{\Phi_G}(\gamma) ) \\ &= \sum_{e \in \gamma} \nu_e - \frac{D}{2} \ell(\gamma) - \omega(G)~ \delta_{\text{m.m.}}(\gamma) \text{ for all non-empty } \gamma \subset G, \end{align*}
where we used Lemma~\ref{lmm:gpMinkowski} and Theorem~\ref{thm:psi_phi_facet} and where $\delta_{\text{m.m.}}(\gamma) = 1$ if $\gamma$ is mass-momentum-spanning and $0$ otherwise. Note that up to the $\delta_{\text{m.m.}}$-term the function $r_G(\gamma)$ is equal to the superficial degree of divergence $\omega(\gamma)$ of a subgraph.

For Euclidean kinematics, the polynomials $\Psi_G$ and $\Phi_G$ have only positive coefficients. Therefore, they are completely non-vanishing on $\P_{>0}^{E-1}$. We can apply Theorem~\ref{thm:secdec_gp} independently of the sign of $\omega(G)$ and find that the parametric integral in eq.~\eqref{eq:parametric} is convergent if $r_G(\gamma) > 0$ for all non-empty proper $\gamma \subsetneq G$ by using Corollary~\ref{crll:gp_inclusion} which implies $\mathcal A \subset \relint \mathcal B$ in this case. In fact, it is sufficient that $r_G(\gamma) > 0$ holds for all proper \emph{motic} subgraphs $\gamma$ as defined in \cite[Definition~3.1]{Brown:2015fyf} for $I_G$ to be convergent. 

As defined in eq.~\eqref{eq:mu_probability} the tropical differential form associated to $I_G$ is
\begin{align*} \mu^\tr_G = \frac{1}{I_G^\tr} \frac{\prod_{e} x_e^{\nu_e}}{\Psi_G^\tr(\bb x)^{D/2}} \left( \frac{\Psi_G^\tr(\bb x)}{\Phi_G^\tr(\bb x)} \right)^{\omega(G)} \Omega, \end{align*}
with an appropriate normalization factor $I^\tr_G$ such that $1 = \int_{\P^{n-1}_{>0}} \mu^\tr_G$. For $\omega(G)=0$ this normalization factor is a certain invariant of the graph $G$ which has been studied by Panzer \cite{Panzer:2019yxl}. This invariant is the \emph{Hepp-bound}. The Hepp-bound is independent of the physical parameters encoded in the masses and external momenta. It mirrors many properties of the \emph{period}, which is given by the integral in eq.~\eqref{eq:parametric} in the same special case $\omega(G) = 0$. The period is another graph invariant which has interesting number theoretical properties \cite{Broadhurst:1995km,Brown:2009ta,Brown:2009rc,Hu:2018liw}.

By Definition~\ref{def:recursionJ}, the normalization factor can be generalized to a subgraph function $J_G : \bb 2^{[E]} \rightarrow \R$ which is determined by the recursion
\begin{align*} J_G(\gamma) = \sum_{e\in \gamma} \frac{J_G(\gamma \setminus e)}{ r_G(\gamma \setminus e) } \text{ for all non-empty } \gamma \subset G \text{ with } J_G(\emptyset) = 1 \text{ and } r_G(\emptyset) = 1. \end{align*}
The actual normalization factor is recovered for $\gamma = G$, i.e.~$I_G^\tr = J_G(G)$
 by Proposition~\ref{prop:Jrspecial}. With a precalculated table of the values $r_G(\gamma)$ and $J_G(\gamma)$ for all $\gamma \subset G$, Algorithm~\ref{alg:gp_sampling} provides an efficient way to sample from the distribution given by the differential form $\mu^\tr_G$ on $\P^{n-1}_{>0}$. Using this sampling algorithm we can obtain estimates for the parametric Feynman integral eq.~\eqref{eq:parametric} by the standard Monte Carlo procedure from Theorem~\ref{thm:montecarlo} or equivalently Algorithm~\ref{alg:tropical_sampling}.

\subsection{Expansions in regularization parameters}
Often not only the integral in eq.~\eqref{eq:parametric} is of interest, but also the Taylor expansions of the parameters $D$ and $\nu_e$ around specific points. Very important is the $\varepsilon$-expansion of the parametric Feynman integral in eq.~\eqref{eq:parametric} in the context of dimensional regularization. Effectively, such an expansion results in integrals of the form
\begin{align} \widetilde{I}_G &= \int_{\P^{E-1}_{> 0}} \frac{\prod_{e} x_e^{\nu_e}}{\Psi_G(\bb x)^{D/2}} \left( \frac{\Psi_G(\bb x)}{\Phi_G(\bb x)} \right)^{\omega(G)} \left(\prod_{e} \log^{k_e} (x_e) \right) \log^s (\Psi_G) \log^t (\Psi_G/\Phi_G) \Omega, \end{align}
for some set of integers $s,t\in \N$ and $k_1,\ldots,k_E \in \N$.%
The estimation of this generalization is also possible using Algorithm~\ref{alg:mu_sample} or Algorithm~\ref{alg:gp_sampling}. Using
\begin{gather*} \widetilde{I}_G = I^\tr_G \int_{\P^{E-1}_{> 0}} \frac{1}{(\Psi_G(\bb x)/\Psi^\tr_G(\bb x))^{D/2}} \left( \frac{\Psi_G(\bb x)/\Psi_G^\tr(\bb x)}{\Phi_G(\bb x)/\Phi_G^\tr(\bb x)} \right)^{\omega(G)} \times \\ \times \left(\prod_{e} \log^{k_e} (x_e) \right) \log^s (\Psi_G) \log^t (\Psi_G/\Phi_G) \mu^\tr, \end{gather*}
gives the desired estimate. A caveat is that the integrand is not bounded anymore, as the logarithms will exhibit singularities at the boundary of the integration domain. This is not a severe problem, as these singularities are square integrable and Theorem~\ref{thm:montecarlo} may still be applied. %
\subsection{Some experimental results}

\begin{table}
\begin{center}
\begin{tabular}{ c | c | c | c | c | r }
$E$ & $\ell(G)$ & $\sigma_I / I$ & samples per second & preprocessing time & RAM \\ \hline\hline
$6$ & $3$ & $0.9$ & $1.1 \cdot 10^{6} \phantom{{}^{-}}/~s$ & $3.0 \cdot 10^{-5}~s$ & 1 KB \\ \hline
$8$ & $4$ & $1.1$ & $7.5 \cdot 10^{5} \phantom{{}^{-}}/~s$ & $1.3 \cdot 10^{-4}~s$ & 4 KB \\ \hline
$10$ & $5$ & $1.3$ & $5.1 \cdot 10^{5} \phantom{{}^{-}}/~s$ & $6.0 \cdot 10^{-4}~s$ & 16 KB \\ \hline
$12$ & $6$ & $1.6$ & $4.1 \cdot 10^{5} \phantom{{}^{-}}/~s$ & $2.7 \cdot 10^{-3}~s$ & 64 KB \\ \hline
$14$ & $7$ & $1.8$ & $3.2 \cdot 10^{5} \phantom{{}^{-}}/~s$ & $1.2 \cdot 10^{-2}~s$ & 256 KB \\ \hline
$16$ & $8$ & $2.1$ & $2.6 \cdot 10^{5} \phantom{{}^{-}}/~s$ & $5.3 \cdot 10^{-2}~s$ & 1 MB \\ \hline
$18$ & $9$ & $2.5$ & $2.1 \cdot 10^{5} \phantom{{}^{-}}/~s$ & $2.3 \cdot 10^{-1}~s$ & 4 MB \\ \hline
$20$ & $10$ & $2.8$ & $1.4 \cdot 10^{5} \phantom{{}^{-}}/~s$ & $1.1 \cdot 10^{0} \phantom{{}^{-}}~s$ & 16 MB \\ \hline
$22$ & $11$ & $3.2$ & $1.0 \cdot 10^{5} \phantom{{}^{-}}/~s$ & $4.7 \cdot 10^{0} \phantom{{}^{-}}~s$ & 64 MB \\ \hline
$24$ & $12$ & $3.7$ & $8.6 \cdot 10^{4} \phantom{{}^{-}}/~s$ & $2.1 \cdot 10^{1} \phantom{{}^{-}}~s$ & 256 MB \\ \hline
$26$ & $13$ & $4.2$ & $6.9 \cdot 10^{4} \phantom{{}^{-}}/~s$ & $9.5 \cdot 10^{1} \phantom{{}^{-}}~s$ & 1 GB \\ \hline
$28$ & $14$ & $4.8$ & $5.9 \cdot 10^{4} \phantom{{}^{-}}/~s$ & $4.4 \cdot 10^{2} \phantom{{}^{-}}~s$ & 4 GB \\ \hline
$30$ & $15$ & $5.3$ & $5.1 \cdot 10^{4} \phantom{{}^{-}}/~s$ & $1.9 \cdot 10^{3} \phantom{{}^{-}}~s$ & 16 GB \\ \hline
$32$ & $16$ & $6.3$ & $4.3 \cdot 10^{4} \phantom{{}^{-}}/~s$ & $8.7 \cdot 10^{3} \phantom{{}^{-}}~s$ & 64 GB \\ \hline
$34$ & $17$ & $7.2$ & $3.6 \cdot 10^{4} \phantom{{}^{-}}/~s$ & $3.9 \cdot 10^{4} \phantom{{}^{-}}~s$ & 256 GB \\ \hline
\hline
\end{tabular}
\end{center}
\caption{Benchmark of Feynman integral evaluations with different numbers of edges.}
\label{tbl:benchmark}
\end{table}
A proof-of-concept \texttt{C++} implementation of this algorithm, which evaluates general Euclidean Feynman integrals, is available on the author's personal web page\footnote{\href{https://michaelborinsky.com}{michaelborinsky.com}} and in the ancillary files to the arXiv version of this article. The algorithm has been tested on various graphs from $\varphi^4$-theory in four dimensions, which have been generated using tools from \cite{Borinsky:2014xwa}. To illustrate the performance of the algorithm a benchmark is given in Table~\ref{tbl:benchmark}. The benchmark has been performed on a single core of an \texttt{AMD EPYC 7702P} processor. The columns $E$ and $\ell(G)$ show the number of edges (equivalently the dimension of the integral $+1$) and the corresponding number of loops of the underlying $\varphi^4$-graph. The column $\sigma_I/I$ gives the relative standard deviation of the samples, i.e.~if $\delta^{-2} \cdot \sigma_I/I$ samples are drawn, then a relative accuracy $\delta$ can be expected from the resulting estimate. Up to this expected accuracy, all obtained estimates are consistent with the available analytic results from \cite{Broadhurst:1995km,Schnetz:2008mp,Brown:2015ztw,Panzer:2016snt,schnetz2018numbers}. The implementation has also been checked using numerical calculations of non-$\varphi^4$ graphs with non-trivial masses and kinematics performed with \texttt{pySecDec} \cite{Borowka:2017idc,Borowka:2018goh}. 

Recall that the algorithm can be applied to arbitrary $D$-dimensional scalar Feynman integrals with arbitrary kinematics in the Euclidean regime and the benchmark results can expected to be representative for the evaluation of all such graphs with the same number of edges. The choice for $\varphi^4$-theory and $D=4$ is practical because much analytic data is available even at high loop orders, which allows for convenient checks of the numerical estimates. 

As can be seen from the table, the number of samples per second decreases slowly with the loop order or equivalently the dimension of the problem. The necessary time for the preprocessing step on the other hand depends exponentially on the dimension.  For example: it takes $2.5$ CPU-seconds to evaluate a graph with $10$ edges and general kinematics up to $\delta= 10^{-3}$ relative accuracy. The necessary time for the preprocessing step of $6.6 \cdot 10^{-4}~s$ is negligible and the memory requirements of $16 \text{ KB}$ insignificant. It takes $20$ CPU-seconds to evaluate a graph with $20$ edges and general kinematics up to the same relative accuracy. The time for the preprocessing step is $1$ second and the memory requirements of $16 \text{ MB}$ are still very manageable.  Similarly, it takes about $2$ CPU-minutes to evaluate a Feynman graph with $30$ edges up to this accuracy, after the preprocessing step has been performed. At this point this preprocessing step unfortunately already takes about $30$ minutes and $16 \text{ GB}$ of RAM are necessary.

The evaluation step of the algorithm is fully parallelizable and the preprocessing step partially. The memory requirements can be reduced in the special case $\omega(G)=0$ or by using a more efficient storage of the relevant constants. The overall picture of exponentially growing memory demands and an exponential time for the preprocessing step will not change without modifying the algorithm significantly.

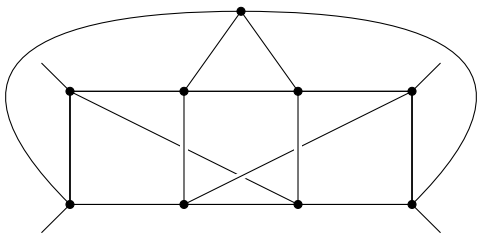
\begin{figure}

    \begin{center}
        \begin{tikzpicture}[scale=1.5] \coordinate (v1) at (0,0); \coordinate (v2) at (1,0); \coordinate (v3) at (2,0); \coordinate (v4) at (3,0); \coordinate (v5) at (0,1); \coordinate (v6) at (1,1); \coordinate (v7) at (2,1); \coordinate (v8) at (3,1); \coordinate (v0) at (1.5,1+0.70710678118 ); \coordinate (i1) at (-.25,-.25); \coordinate (i2) at (-.25,1.25); \coordinate (i3) at (3.25,-.25); \coordinate (i4) at (3.25,1.25); \draw (v5) -- (v3); \draw[preaction={draw, line width=3pt, white}] (v2) -- (v8); \draw[preaction={draw, line width=3pt, white}] (v2) -- (v6); \draw[preaction={draw, line width=3pt, white}] (v3) -- (v7); \draw (v0) -- (v6); \draw (v0) -- (v7); \draw (i1) -- (v1); \draw (i2) -- (v5); \draw (i3) -- (v4); \draw (i4) -- (v8); \draw (v1) -- (v5); \draw (v4) -- (v8); \draw (v1) -- (v2) -- (v3) -- (v4) -- (v8) -- (v7) -- (v6) -- (v5) -- (v1); \draw (v1) -- (v2) -- (v3) -- (v4) -- (v8) -- (v7) -- (v6) -- (v5) -- (v1); \filldraw (v1) circle(1pt); \filldraw (v2) circle(1pt); \filldraw (v3) circle(1pt); \filldraw (v4) circle(1pt); \filldraw (v5) circle(1pt); \filldraw (v6) circle(1pt); \filldraw (v7) circle(1pt); \filldraw (v8) circle(1pt); \filldraw (v0) circle(1pt); \draw (v0) to[out=180,in=135,looseness=2] (v1); \draw (v0) to[out=0,in=45,looseness=2] (v4); \end{tikzpicture}

    \end{center}
    \caption{A $8$-loop $\varphi^4$-graph whose period does not evaluate to a linear combination of multiple zeta values or multiple polylogarithms at roots of unity.} 
    \label{fig:brownschnetz_monster}
\end{figure}
An interesting example of a $\varphi^4$-graph in $D=4$, whose evaluation was not approachable by any previously existing techniques, is the graph in Figure~\ref{fig:brownschnetz_monster}. It is one of the smallest graphs in $\varphi^4$-theory whose period is not a linear combination of multiple zeta values or multiple polylogarithms at roots of unity. This has been proven in \cite[Section~6.2]{Brown:2010bw} for a graph which is equivalent with respect to its period by the \emph{completion} identity \cite{Schnetz:2008mp}.
Sampling $10^{12}$ points in about $24$ hours on $54$-CPU-cores results in the following estimate for the period of this graph,
\begin{align*} I_G &= \int_{\P^{E-1}_{> 0}} \frac{\prod_{e} x_e}{\Psi_G(\bb x)^{2}} \Omega \approx 422.9610 \pm 0.0009 . \end{align*}

\section{Further research directions}
\label{sec:openquestions}
\begin{enumerate}
\item 
\textbf{(Markov chain Monte Carlo based sampling)}
The tropical Monte Carlo algorithms are still very limited in terms of the complexity of the integrals to which they apply, because of the cumbersome preprocessing step that has to be performed for each integral. 
To overcome this bottleneck without relying on special structures of the integrals, it would be necessary to find a more efficient way to sample from $\mu^{\tr}$ than Algorithm~\ref{alg:mu_sample} or Algorithm~\ref{alg:gp_sampling}, while also having access to the normalization factor $I^\tr$. Eventually, one has to settle with a still relatively slow algorithm for this task, as we have a `no-go Theorem' in a special case: if all the numerator polynomials are monomials, i.e.\ $a_i(\bb x) = x_i$, then $I^\tr$ corresponds to the volume of a certain polytope. Computing or approximating the volume of a general $n$-dimensional polytope is a task that cannot be performed deterministically in polynomial time \cite{barany1987computing}. A workaround is to use a non-deterministic algorithm  for both the computation of the normalization factor $I^\tr$ and to obtain samples from $\mu^\tr$. There are many highly advanced Markov chain Monte Carlo algorithms that have been developed to perform exactly this task (see for instance \cite{dyer1991random,lovasz1993random} and the references therein). It is very plausible that adapting these polytope integration and sampling algorithms to our algebraic integral quadrature application should result in the sought after polynomial time algorithm for algebraic and Feynman integral evaluation.

\item \textbf{(Physical integration regions and components of coamoeba)}
The last condition in Theorem~\ref{thm:convergence} is closely related to the \emph{coamoeba} of the set of polynomials $\{b_j\}$. 
If a polynomial $p\in\C[x_1,\ldots,x_n]$ has zero locus $Z_{p} = \{ \bb z \in (\C \setminus \{ 0 \})^n: p(\bb z) = 0 \}$, then the coamoeba of $p$ is the image of $Z_p$ under the coordinate-wise complex $\arg$-function: $\mathcal A_p' = \operatorname{Arg}( Z_p ) \subset [0,2\pi]^n$. The coamoeba is related to the \emph{amoeba} which goes back to Gelfand, Kapranov and Zelevinsky \cite{gkz1994} and has numerous applications in tropical geometry. By a result proven independently by Johansson \cite{johansson2013argument} and Nisse, Sottile \cite{nisse2013}, a polynomial is completely non-vanishing if the origin is not in the closure of its coamoeba $\bb 0 \not \in \bar {\mathcal A}_p'$. In \cite{nilsson2013mellin} it was shown via Cauchy's theorem that the integration cycle $\R^{n-1}_{>0}$ of the integral in eq.~\eqref{eq:integral_euler_mellin} can be replaced with the $\operatorname{Arg}^{-1}(\theta)$ as long as $\theta$ and $\bb 0$ lie in the same connected component of the intersection of the coamoeba of the denominator polynomials. A similar argument works for the projective version of generalized Euler-Mellin integrals which was considered here.

A strikingly reminiscent procedure is necessary while evaluating Feynman integrals with kinematics in Minkowski space. The necessary analytic continuation in this case is governed by the $i \varepsilon$-prescription, which ultimately results from causality and unitarity constraints on the amplitude \cite{Eden:1966dnq}. Formulating this procedure in terms of a canonical choice of a component in the respective coamoeba would result in a canonical analytic continuation procedure in the Minkowski case.
See also 
\cite{delaCruz:2019skx} where related observation regarding parametric Feynman integrals and coamoeba have been made.

\item
\textbf{(Further acceleration of the algorithms by using more structures)}
In the $\omega(G)=0$ case the normalization factor of the $\mu^\tr_G$ distribution for the parametric Euclidean Feynman integral in eq.~\eqref{eq:parametric} reduces to the Hepp-bound studied by Panzer \cite{Panzer:2019yxl}. He gave more efficient ways to compute this normalization factor $I^\tr_G$ which likely can also be used to sample from the $\mu^\tr_G$ distribution more efficiently. Moreover, these more elaborate ways to compute the Hepp-bound can probably be generalized to deal with the $\omega(G) \neq 0$ case using results from Brown~\cite{Brown:2015fyf}.

In a broader sense, an extension of Algorithm~\ref{alg:gp_sampling} beyond the generalized permutahedron case could be possible. 
Especially attractive would be an extension which includes the interesting Minkowski space Feynman integral case. Further analysis of the relevant structures for the tropical geometric framework, starting for instance with the explicit counterexample in \cite[Section~2.4]{Smirnov:2012gma}, could lead to an appropriate refinement of the braid arrangement fan. Such a refinement could lead to a direct generalization of the generalized permutahedron sampling Algorithm~\ref{alg:gp_sampling}, which would make the integration of high dimensional Feynman integrals (i.e.~with $\sim 30$ edges) also possible in the Minkowski regime.
\item 
\textbf{(BPHZ renormalization)}
As mentioned above, Feynman integrals as the one in eq.~\eqref{eq:parametric} with non-integrable singularities are often interesting. A common approach to deal with these singularities is to subject the integral to an analytic continuation procedure before any numerical integration is performed. Ultimately, these singularities have a well-studied physical origin and are handled via \emph{renormalization}. The momentum BPHZ renormalization scheme takes care of these singularities \emph{before} any integration is performed. This renormalization scheme can be implemented on the level of the parametric integrand \cite{Brown:2011pj}. Such an implementation would make the analytical continuation step in the $\nu_e$ edge weights and the dimension $D$ unnecessary.
\item
\textbf{(Estimates for large loop order $\beta$-functions)}
The estimates that can be obtained using the proof-of-concept implementation of Feynman graph integrals up to loop order $17$ in $\varphi^4$-theory can immediately be used for a numerical estimation of the $\beta$-function up to this loop order. This has phenomenological applications for the calculation of critical exponents for various complex systems and even works without the need for an analytic continuation as it has been observed that the non-primitive contributions to the $\beta$-function become negligible with sufficiently large loop order. The approach can be further amplified by making use of the observed Hepp-bound - period correlation which has been applied by Panzer and Kompaniets \cite{Kompaniets:2017yct} to obtain estimates of the $\varphi^4$-theory $\beta$-function up to order $13$. 

Even if it is not possible to evaluate all necessary Feynman diagrams individually for the respective loop order, a numerical approach could be sufficient to gain enough insights on the distribution of the value of these integrals. If such statistical knowledge is available, the numbers of (renormalized) Feynman diagrams are sufficient to extrapolate values for the entire $\beta$-function contribution \cite{Borinsky:2017hkb}.

Such an approach naturally extends to a question for the inherently large-order regime: suppose that $G$ is a random 1PI $\varphi^4$-graph without subdivergences and $\ell(G)$ loops. Is there a limiting distribution 
\begin{align*} \lim_{\ell(G) \rightarrow \infty} C_{\ell(G)}^{-1} \int_{\P^{E-1}_{> 0}} \frac{\prod_{e} x_e^{\nu_e}}{\Psi_G(\bb x)^{D/2}} \left( \frac{\Psi_G(\bb x)}{\Phi_G(\bb x)} \right)^{\omega(G)} \Omega, \end{align*}
with an appropriate normalization constant $C_{L}$ for each loop order $L=\ell(G)$ and if yes, what does it look like?
The analysis \cite{deCalan:1981szv} gives some positive indication for the existence of such a distribution.

An overall normalization constant for such a distribution, which normalizes its expectation value to one, can be calculated using instanton methods \cite{McKane:2018ocs} and renormalized graph counting \cite{Borinsky:2017hkb}, as was pointed out by Panzer~\cite{panzer2020personal}:
\begin{align*} C_{L} = \frac{4 e^{-3\gamma_E}}{\sqrt{2 \pi} A^6} L^{\frac52} \left(\frac{3}{2}\right)^{L+3}, \end{align*}
where $\gamma_E$ is the Euler--Mascheroni constant and $A$ is the Glaisher--Kinkelin constant.

An exhaustive statistical analysis using the algorithms from this article should give further indication for or against the existence of such a limit. A combination with analytic combinatorial methods for Dyson-Schwinger equations might lead to an explicit form of a limit distribution of Feynman integrals \cite{Kreimer:2006ua,Kreimer:2006gm,Courtiel:2019dnq}.
\item
\textbf{(Phase-space integration)}
Simple \emph{phase space integrals}, which are another type of integrals necessary for particle physics phenomenology, also fall under the category of integrals in eq.~\eqref{eq:integral}. For more elaborate phase space integrals more complicated non-simplicial integration domains are necessary. It is possible that the algorithms discussed in this article may be extended to these more complicated domains. Writing the phase space integrals in terms of kinematic variables as in \cite{Gehrmann-DeRidder:2003pne} and using a geometric subtraction scheme for the infrared singularities \cite{Herzog:2018ily} could be instrumental for this extension. 
\item
\textbf{(Tropical sampling applied to sums of Feynman diagrams)}
The general tropical sampling algorithm described in Section~\ref{sec:trop_sampling} is made possible by a well-calculated emancipation from the rigid concept of \emph{sectors} as parts of the integral, which each have to be attacked individually. 

Following a line of thought from \cite{Arkani-Hamed:2017tmz}, we can say that a similar but much stronger bias exists on the level of the amplitude. The time-honoured approach to amplitude calculation is to write it as a sum over Feynman graphs with the same number of loops $L$ (pictorially in disregard of renormalization and the explicit form of the integrals),
\begin{align*} A_{L} = \sum_{\ell(G) = L} \frac{1}{|\Aut G|} \int \left( \ldots \right) \Omega, \end{align*}
and evaluate the indicated integrals one by one.
There are promising indications that there is a superior structure which can be `triangulated' into Feynman integrals in an appropriate sense yielding a sum as the one above, similar to the sector decomposition approach, where an individual integral is decomposed in terms of the triangulation of the respective normal fan. 

An especially suggestive candidate for such a superior object in the case of scalar quantum field theories is Outer space \cite{Culler1986ModuliOG} and its quotient formed under the action of $\operatorname{Out}(F_n)$ which is the \emph{moduli space of graphs}. For instance, unitarity and branch cut properties of Feynman integrals can be understood using Outer space \cite{bloch2015cutkosky,Kreimer:2018kah,Berghoff:2017dyq}. This space can be seen as a tropical analogue of Teichmüller space \cite{chan2013tropical} and the moduli space of curves, which holds a similar superior role in string theory. Recently, quantum field theory inspired techniques have been successfully applied in the theory of Outer space \cite{Borinsky:2019rtu}.

A problem to overcome for such an approach are the UV-divergences that naturally appear in renormalizable QFT calculations. It is well-known how such divergences can be handled both on the amplitude or on a per integral level \cite{Collins:1984xc} and also mathematically these divergences are quite well understood, even in the large-order regime \cite{Borinsky:2017hkb,borinsky2018graphs}. These divergences would also appear in a geometric setting for the amplitude and dealing with them would mean to work on a certain compactification of Outer space and the moduli space of graphs. One such compactification has been constructed by Berghoff \cite{Berghoff:2017dyq} (see also \cite{Berghoff:2020bug}), which might be usable for the numerical evaluation of amplitudes.
\item 
\textbf{(Quasi Monte Carlo)}
State of the art implementations for numerical Feynman integral integration employ \emph{quasi Monte Carlo} methods \cite{niederreiter1992random} for the actual integration of the sector integrals instead of traditional Monte Carlo methods. This has the simple and obvious advantage of a significantly increased rate of convergence. The disadvantage of the quasi Monte Carlo approach is that it is mathematically much more challenging to handle. It is plausible that the algorithms introduced in this article can be accelerated using quasi Monte Carlo methods. A challenge will be the handling of the mixture of discrete and continuous probability distributions in the sampling algorithms from Sections~\ref{sec:trop_sampling} and \ref{sec:genperm}.
\item 
\textbf{(Systematic tropical expansions)}
The tropical approximation $p^\tr$ can be interpreted as a certain limit as shown in Section~\ref{sec:trop_approx}. It is natural to ask if one can interpret $p^\tr$ as the `zeroth' order in a systematic expansion and it is plausible that a systematic improvement of the technique can be obtained this way. The ultimate aim of such investigations would be an efficient approximation scheme that gets by without a final Monte Carlo step and immediately yields a deterministic result.

\end{enumerate}

\providecommand\noopsort[1]{}

\providecommand{\href}[2]{#2}\begingroup\raggedright\endgroup
\end{document}